\RequirePackage{fix-cm}
\documentclass[smallcondensed]{svjour3}     
\smartqed  

\usepackage{graphicx}
\usepackage{
  amsmath,
  amssymb,
  graphicx
}

\usepackage{bbold}
\usepackage{mathtools}

\newcommand\restr[2]{\ensuremath{\left.#1\right|_{#2}}}

\DeclareMathOperator{\poly}{poly}

\usepackage[table]{xcolor}

\usepackage{algorithm}
\usepackage[noend]{algpseudocode}
\usepackage{tikz}
\usepackage{tkz-berge}
\usetikzlibrary{external,backgrounds}
\usetikzlibrary{positioning}
\usepackage{mathtools}
\usepackage{subcaption}
\usepackage{bbold}
\usepackage{xfrac}
\usepackage{ifthen}
\usetikzlibrary{decorations.markings}
\usepackage{diagbox}
\usepackage{tabularx}
\usepackage{array}
\usepackage[export]{adjustbox}
\usepgfmodule{nonlineartransformations}

\algnewcommand\And{\textbf{and}}
\algnewcommand\Or{\textbf{or}}

\newcolumntype{Y}{>{\centering\arraybackslash}X}

\pgfdeclarelayer{background}
\pgfdeclarelayer{foreground}
\pgfsetlayers{background,main,foreground}

\usetikzlibrary{decorations.pathreplacing,calc}
\tikzset{draw half paths/.style 2 args={%
  decoration={show path construction,
    lineto code={
      \draw [#1] (\tikzinputsegmentfirst) --
         ($(\tikzinputsegmentfirst)!0.2!(\tikzinputsegmentlast)$);
      \draw [#2] ($(\tikzinputsegmentfirst)!0.2!(\tikzinputsegmentlast)$)
        -- (\tikzinputsegmentlast);
    }
  }, decorate
}}

\tikzset{
    my box/.style = {
        , line cap = round
        , line join = round
    }
  }

\newcommand{\highlight}[3]{
  \path [my box, line width = #1, draw = #2, transparency group, opacity=1] #3;
}

\usetikzlibrary{snakes}
\tikzset{wavy/.style={decorate,decoration={snake,amplitude=.4mm,segment length=2mm,post length=0mm,pre length=0mm},line width=.5}}

\tikzset{myDotted/.style={line width=1,dash pattern={on 1pt off 3pt}}}
\tikzset{myDashed/.style={line width=1,dash pattern={on 3pt off 3pt}}}
\tikzset{myDashDotted/.style={line width=1,dash pattern={on 3pt off 3pt on 1pt off 3pt}}}

\SetVertexNormal[FillColor=gray!25]

\makeatletter
\renewcommand*{\write@math}[3]{%
            \pgfmathtruncatemacro{\printindex}{#3+1}
            \Vertex[x = #1,y = #2,%
                    L = \cmdGR@cl@prefix\grMathSep{\printindex}]{\cmdGR@cl@prefix#3}}

\def\R{\mathbb R}

\def\Q{\mathbb Q}
\def\Z{\mathbb Z}
\def\N{\mathbb N}
\def\1{\mathbb 1}

\def\mc{\mathcal}

\DeclareMathOperator*{\argmax}{arg\,max}
\DeclareMathOperator*{\argmin}{arg\,min}

\tikzset{v 0/.style={fill=white}, v 1/.style={fill=black!40},
   venn path 1/.style={insert path={
     (90:1/sqrt 3) arc (60:120:1) arc (180:0:1) arc (60:120:1) -- cycle }},
   venn path 2/.style={insert path={
     (90:1/sqrt 3) arc (120:180:1) arc (240:180:1) arc (120:60:1) -- cycle }},
   venn path 3/.style={insert path={
     (90:1/sqrt 3) arc (120:180:1) arc (240:300:1) arc (0:60:1) -- cycle }},
   pics/venn 3/.style args={#1#2#3#4#5#6#7#8}{code={%
     \fill [v #1] (-1.5-.5, -1.3-.5) rectangle (1.5+.5,1.6+.5);
     \fill [v #2, rotate=240, venn path 1];
     \fill [v #3, rotate=120, venn path 1];
     \fill [v #4, venn path 1];
     \fill [v #5, rotate=240, venn path 2];
     \fill [v #6, rotate=120, venn path 2];
     \fill [v #7, venn path 2];
     \fill [v #8, venn path 3];
     \draw [line width = .05mm] (90:1/sqrt 3) circle [radius=1] (210:1/sqrt 3) circle [radius=1]
     (330:1/sqrt 3) circle [radius=1] (-1.5-.5, -1.3-.5) rectangle (1.5+.5,1.6+.5);
     \node at (0,1.65){\Large $E_1$};
     \begin{scope}[rotate=120]
       \node at (0,1.65){\Large $E_2$};
     \end{scope}
     \begin{scope}[rotate=-120]
       \node at (0,1.65){\Large $E_3$};
     \end{scope}
}}}
\pgfmathsetbasenumberlength{8}

\tikzset{every venn/.style={x=1em, y=1em, baseline=-.666ex,
    v 1/.style={fill=gray}}}

\makeatletter
\def\mytransformation{%
\pgfmathsetmacro{\myX}{\pgf@x}
\pgfmathsetmacro{\myY}{\pgf@y}
\setlength{\pgf@x}{\myX pt}
\setlength{\pgf@y}{\myY pt}
}
\makeatother

\newcommand\mtiny[1]{\mbox{\tiny\ensuremath{#1}}}

\begin{document}
\title{The Simultaneous Assignment Problem}
\author{P\'eter Madarasi}

\institute{P. Madarasi \at
Department of Operations Research, ELTE E\"otv\"os Lor\'and University, and the ELKH-ELTE Egerv\'ary Research Group on Combinatorial Optimization, E\"otv\"os Lor\'and Research Network (ELKH), P\'azm\'any P\'eter s\'et\'any 1/C, 1117 Budapest, Hungary. \email{madarasip@staff.elte.hu}
}

\date{}

\maketitle

\begin{abstract}
  This paper introduces the \emph{simultaneous assignment problem}.
  Let us given a graph with a weight and a capacity function on its edges, and a set of its subgraphs along with a degree upper bound function for each of them.
  We are also given a laminar system on the node set with an upper bound on the degree-sum of the nodes in each member of the system.
  Our goal is to assign each edge a non-negative integer below its capacity such that the total weight is maximized, the degrees in each subgraph are below the degree upper bound associated with the subgraph, and the degree-sum bound is respected in each member of the laminar~system.

  We identify special cases when the problem can be shown to be solvable in polynomial time.
  One of these cases is a common generalization of the \emph{hierarchical $b$-matching problem} and the \emph{laminar matchoid problem}.
  This implies that both problems can be solved efficiently in the weighted, capacitated case even if both lower and upper bounds are present  --- generalizing the previous polynomial-time algorithms.
  The problem is also solvable for trees provided that the laminar system is empty and a natural assumption holds for the subgraphs.

  The general problem, however, is shown to be APX-hard in the unweighted case, which also
  implies that the \emph{weighted distance matching problem} is APX-hard.
  Furthermore, we prove that the approximation guarantee of any polynomial-time algorithm must increase asymptotically linearly in the number of the given subgraphs, unless P=NP.
  We give a generic framework for deriving approximation algorithms, which can be applied to a wide range of problems.
  As an application to our problem, a constant-approximation algorithm is derived when the number of the given subgraphs is a constant.
  The approximation guarantee is the same as the integrality gap of a strengthened LP-relaxation when the number of the subgraphs is small.
  Furthermore, improved approximation algorithms are given when the degree bounds are uniform or the graph is sparse.
\end{abstract}

\newpage

\section{Introduction}\label{sec:simultan:intro}
In the \emph{simultaneous assignment problem}, we are given a graph with a weight and a capacity function on its edges and a set of its subgraphs along with a degree upper bound function for each of them.
We are also given a laminar system on the node set with an upper bound on the degree-sum of the nodes in each member of the system.
Our goal is to assign each edge a non-negative integer below its capacity such that the total weight is maximized, the degrees in each subgraph are below the degree upper bound associated with the subgraph, and the degree-sum bound is respected in each member of the laminar system.

More precisely, given are 1) a loop-free graph $G=(V,E)$, 2) a capacity function $c:E\to\Z_+$, 3) a set~$\mathcal{H}$ of the subgraphs of $G$ along with a function $b_H:V_H\to\Z_+$ for each subgraph $H=(V_H,E_H)\in\mathcal{H}$ and 4) a laminar system $\mc L$ on the nodes of $G$ with a function $g:\mc L\to\Z_+$.
An integer vector $x\in\Z_+^E$ is a \emph{simultaneous assignment} if 1) $x\leq c$, 2) $x$ is a $c$-capacitated $b_H$-matching in $H$ for all $H\in\mc H$ and 3) $\sum_{v\in L}\sum_{e\in\Delta_G(v)}x_e\leq g(L)$ holds for all $L\in\mc L$, where $\Delta_G(v)$ denotes the set of edges incident to node $v$.
In the weighted version of the problem, $\sum_{e\in E}w_ex_e$ is to be maximized for a given weight function $w:E\to\R_+$.
The case $w\equiv 1$ will be referred to as the \emph{unweighted} problem.
We will see that this cumbersome problem includes a number of natural, interesting special cases --- some of which are solvable in polynomial time.

Note that constraint 3) poses an upper bound on the $x$-degree sum in each $L\in\mc L$ --- this means that the $x$-values of the edges induced by $L$ count twice while those of the non-induced incident edges once.
Constraints 3) will be referred to as the \emph{degree-sum constraints}.
The problem is called \emph{uncapacitated} when $c\equiv\infty$.
Without loss of generality, we assume that no subgraph in $\mc H$ contains isolated nodes.

The integer solutions of the following linear program are, by definition, the feasible simultaneous assignments.
\begin{subequations}
  \begin{align}\label{lp:simultan:degLp}
    \tag{LP1}
    \max&\sum_{st\in E}w_{st}x_{st}\\
    \mbox{s.t.}\quad\quad\quad\quad\quad&&\nonumber\\
    x&\in\R_+^{E}&\label{eq:simultan:degLp:nonNeg}\\
    x_{e} &\leq c_{e} &\forall e\in E\label{eq:simultan:degLp:cap}\\
    \sum_{e\in\Delta_H(v)} x_{e} &\leq b_H(v) &\forall H\in\mc H\ \forall v\in V_H\label{eq:simultan:degLp:deg}\\
    \sum_{v\in L}\sum_{e\in\Delta_G(v)}x_e &\leq g(L) &\forall L\in\mc L\label{eq:simultan:degLp:card}
  \end{align}
\end{subequations}
Indeed, all feasible integer solutions to~(\ref{lp:simultan:degLp}) are non-negative by~(\ref{eq:simultan:degLp:nonNeg}) and respect the capacity constraints by~(\ref{eq:simultan:degLp:cap}).
The degree constraints in each subgraph $H\in\mc H$ and the degree-sum constraint for $L\in\mc L$ hold by~(\ref{eq:simultan:degLp:deg}) and~(\ref{eq:simultan:degLp:card}), respectively.

\paragraph{\normalfont\textbf{Motivation}}
Imagine $k$ consecutive events taking place in the same hall and a set of attendees who want to buy tickets.
Given the event and the seat, we know how much a ticket costs.
Each customer provides the list of seats that would suit him/her and also selects which of the $k$ events (possibly more than one) they want to attend.
Our goal is to assign customers to seats such that the total income is maximized and if somebody wanted to attend multiple events, then he/she must be either completely refused or seated to the same place for all the events.
This scenario can be modeled as a simultaneous assignment problem as follows.
Define a bipartite graph $G=(S,T;E)$, where $S$ corresponds to the customers and $T$ to the seats.
For $s\in S$ and $t\in T$, add $st$ to $E$ if customer $s$ likes seat $t$, and let $w(st)$ be the income if $s$ is seated to $t$.
Let $S_i$ denote the set of customers who want to attend event $i$, and let $\mc H = \{H_1,\dots,H_k\}$, where $H_i$ is the subgraph of $G$ induced by node set $S_i\cup T$.
Let $c\equiv 1$, $\mc L=\emptyset$ and $b_{H_i}\equiv 1$ for $i\in\{1,\dots,k\}$.
By construction, there is a one-to-one correspondence between feasible customer-seat assignments and feasible simultaneous assignments.

It is quite natural to require these constraints only for intervals of events, that is, if a customer skips some of the events, then he/she may be seated to a new place when he/she arrives back.
Observe that a participant $s$ leaving the hall at some point can be replaced with new dummy participants each of whom attends exactly one of the intervals of the events selected by $s$.
That is, one can assume that each customer participates in an interval of events.
This special case will be investigated in Section~\ref{sec:simultan:trees}.

The simultaneous assignment problem can be interpreted as an investment problem, where participants may also invest into each other.
Let $V$ denote the set of instruments (or participants), and let $E$ denote the possible investments (parallel edges are allowed).
Each investment $uv\in E$ has an associated price $w(uv)$ and a bound on the number of such investment $c(uv)$.
Each $v\in V$ can pose a bound $b_F$ on $\sum_{e\in F}x_e$ for any subset $F$ of the edges incident to $v$, which defines the subgraphs in $\mc H$ --- this is where the sense of the investment can be taken into account.
In addition, a laminar system $\mc L$ on the instruments can be also given along with a bound on $\sum_{v\in L}\sum_{e\in\Delta_G(v)}x_e$ for each $L\in\mc L$, which is the total money flow involving some of the instruments in $L$.
Our goal is to find an assignment of maximum total weight respecting all bounds above.
Note that one may unite some of the subgraphs in $\mc H$.
For example, if certain investments (edges) carry extreme risk, then many of the investors may pose a degree bound on those very edges, all of which bounds can be posed in a single subgraph in $\mc H$ consisting of the high-risk investments.

\paragraph{\normalfont\textbf{Previous work}} 
In the special case when $\mc H=\{G\}$ and $\mc L = \emptyset$, one obtains the usual \emph{weighted capacitated $b$-matching problem}, where $b=b_G$~\cite{anstee1987polynomialBMatching}.
Another way to obtain this problem is when $\mc H=\emptyset$ and $\mc L = \{\{v\}:v\in V\}$, where $b(v)=g(\{v\})$ for all $v\in V$.

For the \emph{$\ell$-matchoid problem}, there exists an FPT algorithm parameterized by $\ell$ and the size of the solution~\cite{huang2021fptalgorithms}.
This immediately implies that for the simultaneous assignment problem with $\mc L=\emptyset$ and $c\equiv 1$, there exists an FPT algorithm parameterized by the size of the solution and the size of $\mc H$.

One can show that uncapacitated simultaneous assignments form a \emph{$(2|\mc H|+1)$-extendible system}, hence the greedy algorithm is a $(2|\mc H|+1)$-approximation algorithm~\cite{MestreExtendible}.
This result is not hard to extend to the capacitated version, hence one gets that there exists a $(2|\mc H|+1)$-approximation algorithm for the simultaneous assignment problem.

In the \emph{double matching problem}, we are given a bipartite graph $G=(S,T;E)$ and $S_1,S_2\subseteq S$ such that ${S_1\cup S_2=S}$.
It is NP-complete to decide whether there exists $M\subseteq E$ for which $|M|=|S|$ and both $M\cap E_1$ and $M\cap E_2$ are matchings, where $E_i$ denotes the edges induced by $T$ and $S_i$ for $i\in\{1,2\}$~\cite{madarasi2021matchings}.
The double matching problem is a special case of the simultaneous assignment problem, and this implies that it is NP-complete to decide whether a simultaneous assignment satisfying constraints~(\ref{eq:simultan:degLp:deg}) with equality exists, even if $\mc L=\emptyset$.

As a direct application of the bounded-violation algorithms given for the \emph{upper bounded degree $g$-polymatroid element problem} described in~\cite{berczi2019degreebounded}, one can find a vector $z\in\Z^E$ in polynomial time such that $wz$ is at least the weight of the optimal simultaneous assignment, and it satisfies constraints~(\ref{eq:simultan:degLp:nonNeg})~and~(\ref{eq:simultan:degLp:cap}), but it may violate constraints~(\ref{eq:simultan:degLp:deg}) by an additive factor of at most $(2|\mc H|-1)$, provided that $\mc L=\emptyset$.

\paragraph{\normalfont\textbf{Our results}}
The special case when $\mc H=\emptyset$ corresponds to the so-called \emph{weighted hierarchical $b$-matching problem}.
This problem was introduced in~\cite{hierarchicalBMatching}, where a strongly polynomial-time algorithm was given for the unweighted case.
Answering an open question from the same paper, our results in Section~\ref{sec:simultan:locLamSubGraphs} imply that the weighted version of the problem can be solved in strongly polynomial time as well.

If $\mc L=\emptyset$ and $\mc H$ is such that the subgraphs in $\mc H$ restricted to the edges incident to $v$ form a laminar system for each node $v\in V$, then we get back the \emph{laminar matchoid problem}~\cite{jenkynsMatchoid}.
In~\cite{Kaparis2014OnLM}, the laminar matchoid problem was solved in polynomial time when the so-called similarity condition holds, that is, the components of $b$ and $c$ are polynomial in the size of $V$.
Our results in Section~\ref{sec:simultan:locLamSubGraphs} also imply that the problem is solvable in strongly polynomial time even if the similarity condition does not hold.

In fact, Section~\ref{sec:simultan:locLamSubGraphs} solves the simultaneous assignment problem in strongly polynomial time when $\mc H$ is such that the subgraphs in $\mc H$ restricted to the edges incident to $v$ form a laminar system for each node $v\in V$.
This can be seen as a common generalization of the (weighted, capacitated) hierarchical $b$-matching and the laminar matchoid problems, in which the laminar matchoid problem subject to the degree-sum (or hierarchical) constraints is to be solved.
This approach settles the weighted, capacitated version of this common generalization even if both lower and upper bounds are given on the degrees in the subgraphs in $\mc H$ and on the degree-sums in the members of $\mc L$ --- generalizing the hierarchical $b$-matching and the laminar matchoid problems with the presence of capacities and both lower and upper bounds.

We show in Section~\ref{sec:simultan:trees} that the simultaneous assignment problem can be solved for trees when $\mc L=\emptyset$ and the so-called local-interval property holds --- the latter corresponds to the special case of the first motivation above when each customer is supposed to participate in an interval of the events.

Section~\ref{sec:simultan:hardness} proves the NP-hardness of $\alpha$-approximating the unweighted problem on bipartite graphs in two special cases for small enough constant $\alpha$: 1) the size of $\mc H$ is two and all connected components of $G$ are claws 2) each subgraph in $\mc H$ consists of two edges, the size of all members of $\mc L$ is at most two and all connected components of $G$ are claws.
Furthermore, we also show that the weighted problem is hard to approximate when $|\mc H|=2$ and $\mc L$ is empty.
Section~\ref{sec:simultan:hardness} also shows that the approximation guarantee of any polynomial-time approximation algorithm must grow asymptotically (essentially) linearly in the size of $\mc H$.
Section~\ref{sec:simultan:hardness:conseqToDM} shows that these results imply the inapproximability of the \emph{weighted non-cyclic} and the \emph{unweighted cyclic distance matching problems}.

Then, Section~\ref{sec:simultan:apx} introduces the concept of $(m,\ell)$-covers and gives a general framework for deriving approximation algorithms, which can be applied to any problem in which one has an efficiently solvable special case with which every instance of the problem can be covered ``equitably''.
The rest of Section~\ref{sec:simultan:apx} applies this approach to the simultaneous assignment problem.
First, Section~\ref{sec:simultan:apx:constAPX} gives an approximation algorithm when the size of $\mc H$ is small, and gives a bound on the integrality gap of a strengthened LP-relaxation.
In Section~\ref{sec:simultan:apx:uniformBounds}, we give an improved algorithm for the uniform case, that is, when the coordinates of all the degree bounds $b_{H}$ are the same.
Section~\ref{sec:simultan:apx:sparse} gives further improved approximation algorithms for special graph classes, for example, for pseudo-trees and sparse graphs.

Finally, Section~\ref{sec:simultan:open} concludes the paper with open questions.

\paragraph{\normalfont\textbf{Notation}}
Let $N_G(v)$ denote the set of the neighbors of $v$.
For a subset $X$ of the nodes, $\Delta_G(X)$ denotes the union of the edges incident to the nodes in $X$.
We use $\deg_G(v)$ to denote the degree of node $v$ in $G$.
The maximum of the empty set is $-\infty$ by definition.
Given a function $f:A\to B$, both $f(a)$ and $f_a$ denote the value $f$ assigns to $a\in A$, and let $f(X)=\sum_{a\in X}f(a)$ for $X\subseteq A$.
Let $\chi_Z$ denote the characteristic vector of set $Z$, that is, $\chi_Z(y)=1$ if $y\in Z$, and $0$ otherwise.
Occasionally, the braces around sets consisting of a single element are omitted, for example, $\Delta_G(\{v\})=\Delta_G(v)$ for $v\in V$.
The power set of a set $X$ is denoted by $2^X$.
Let $\N$ and $\Z_+$ denote the sets of positive and non-negative integers, respectively.

\section{Tractable Cases}
This section introduces some polynomial-time solvable special cases of the simultaneous assignment problem.

\subsection{Locally Laminar Subgraphs}\label{sec:simultan:locLamSubGraphs}
A set system $\mc F$ is \emph{laminar} if, for any two members $X,X'\in\mc F$, either $X\subseteq X'$, $X'\subseteq X$ or $X\cap X'=\emptyset$ holds.
A set $\mc H$ of the subgraphs of $G$ is \emph{laminar} if the edge sets of the subgraphs in $\mc H$ form a laminar system.
We say that $\mc H$ is \emph{locally laminar} if, for each node $v\in V$, the subgraphs in $\mc H$ restricted to~$\Delta_G(v)$ form a laminar system, that is, $\mc F_v=\{\Delta_H(v) : H\in\mc H\}$ is laminar for all $v\in V$.

By definition, if $\mc H$ is laminar, then it is locally laminar.
The reverse direction, however, does not necessarily hold.
Figure~\ref{fig:simultan:locLamEg} gives a locally laminar example with $\mc H=\{H_1,H_2\}$, where $H_1$ and $H_2$ are the two highlighted subgraphs, respectively.
In this case, $\mc H$ is locally laminar, but it is not laminar.

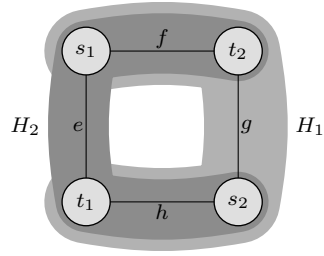
\begin{figure}
  \centering
  \begin{tikzpicture}[scale=1]
    \SetVertexMath
    dot/.style = {circle, draw, semithick,
      inner sep=0pt, minimum size=3pt,
      node contents={}},
    \tikzset{edge/.style = {->,> = latex'}}
    \Vertex[x=-1, y=1,L={s_1}]{s_1}
    \Vertex[x=1, y=1,L={t_2}]{t_2}
    \Vertex[x=-1, y=-1,L={t_1}]{t_1}
    \Vertex[x=1, y=-1,L={s_2}]{s_2}
    \draw[] (s_1) -- (t_2) node [midway, above=-2pt, opacity=100] {$f$};
    \draw[] (s_2) -- (t_2) node [midway, right=-2pt, opacity=100] {$g$};
    \draw[] (s_2) -- (t_1) node [midway, below=-2pt, opacity=100] {$h$};
    \draw[] (s_1) -- (t_1) node [midway, left=-2pt, opacity=100] {$e$};

    \begin{pgfonlayer}{background}
      \begin{scope}[opacity=.8,transparency group]
        \highlight{11mm}{black!30}{(s_1.center) to [bend left =10] (t_2.center) to [bend left =10] (s_2.center) to [bend left =10] (t_1.center)}
      \end{scope}
      \begin{scope}[opacity=.8,transparency group]
        \highlight{8mm}{black!45}{(t_2.center) to [bend right =10] (s_1.center) to[bend right =10] (t_1.center) to[bend right =10] (s_2.center)}
      \end{scope}
      \begin{pgfonlayer}{background}
      \end{pgfonlayer}
    \end{pgfonlayer}
    \node at (-1.8,0){$H_2$};
    \node at (+1.95,0){$H_1$};
  \end{tikzpicture}
  \caption{A locally laminar setting which is not laminar.
    Here $H_1$ and $H_2$ are the light and dark subgraphs, respectively.}\label{fig:simultan:locLamEg}
\end{figure}

In the rest of this section, $\mc H$ is assumed to be locally laminar.
A polynomial-time algorithm will be given to find a maximum-weight simultaneous assignment, in addition, the description of the polyhedron of feasible solutions will be derived.
First, the case of bipartite graphs is considered.

\subsubsection{Bipartite Graphs with Special Degree-sum Constraints}\label{sec:simultan:laminarBP}
Let $G=(S,T;E)$ be a bipartite graph, and let $V=S\cup T$.
In this section, we assume that no member of $\mc L\subseteq 2^{V} $ intersects both $S$ and $T$.

\begin{theorem}\label{thm:simultan:bpLocLam}
  Suppose that $G$ is bipartite, $\mc H$ is locally laminar, and also that each member of $\mc L$ is a subset of either $S$ or $T$.
  Then, the simultaneous assignment problem can be solved in strongly polynomial time and their polyhedron is described by~(\ref{lp:simultan:degLp}).
  Furthermore, the problem remains tractable even if some of the inequalities in~(\ref{lp:simultan:degLp}) are tightened to equality.
\end{theorem}
\begin{proof}
  In this case, the matrix of~(\ref{lp:simultan:degLp}) can be obtained as the transposes of the incidence matrices of two laminar systems written under each other, therefore it is a network matrix~\cite[Page 151]{AF11}.
  As the right-hand side is integer,~(\ref{lp:simultan:degLp}) describes the polyhedron of simultaneous assignments --- even if some of the inequalities are tightened to equality~\cite[Page 152]{AF11}.
  This implies that the problem can be reduced to the \emph{weighted circulation problem}~\cite[Page 158]{AF11}, and hence it is solvable in strongly polynomial time~\cite{Tardos1985}.
  \qed
\end{proof}

Note that Theorem~\ref{thm:simultan:bpLocLam} does not hold for arbitrary laminar $\mc L$.
Figure~\ref{fig:simultan:bpLamConditionNec} shows an example, where a member of $\mc L$ intersects both $S$ and $T$.
In this case, (\ref{lp:simultan:degLp}) is not integer, since $x\equiv \frac{1}{2}$ is the optimal solution with objective value $\frac{3}{2}$.
The next section derives the description of the polyhedron of the degree-sum-constrained simultaneous assignments for arbitrary laminar $\mc L$ and non-bipartite graphs.

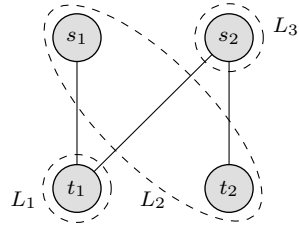
\begin{figure}
  \centering
  \begin{tikzpicture}[scale=1]
    \SetVertexMath
    dot/.style = {circle, draw, semithick,
      inner sep=0pt, minimum size=3pt,
      node contents={}},
    \tikzset{edge/.style = {->,> = latex'}}
    \Vertex[x=-1, y=1,L={s_1}]{s_1}
    \Vertex[x=1, y=1,L={s_2}]{s_2}
    \Vertex[x=-1, y=-1,L={t_1}]{t_1}
    \Vertex[x=1, y=-1,L={t_2}]{t_2}
    \draw[] (t_2) -- (s_2);
    \draw[] (s_2) -- (t_1);
    \draw[] (s_1) -- (t_1);
    \begin{scope}[rotate=45]
      \draw[dashed] (0,0) ellipse (.7cm and 1.9cm);
    \end{scope}
    \draw[dashed] (1,1) ellipse (.45cm);
    \draw[dashed] (-1,-1) ellipse (.45cm);
    \node at (0,-1.15){$L_2$};
    \node at (-1.7,-1.15){$L_1$};
    \node at (1.75,1.15){$L_3$};
  \end{tikzpicture}
  \caption{For $\mc L=\{\{t_1\},\{s_1,t_2\},\{s_2\}\}$, $\mc H =\emptyset$, $g\equiv 1$, $w\equiv 1$ and $c\equiv\infty$, the optimal fractional solution is $x\equiv \frac{1}{2}$, which shows that (\ref{lp:simultan:degLp}) is not integer.}\label{fig:simultan:bpLamConditionNec}
\end{figure}

\subsubsection{Non-bipartite Graphs}
Throughout this section, assume that $\mc H$ is locally laminar.
In this section, a polynomial-time algorithm is given to solve the weighted capacitated degree-sum-constrained simultaneous assignment problem under this assumption.
The description of the polyhedron of feasible simultaneous assignments will be derived as well.
As we have already seen,~(\ref{lp:simultan:degLp}) does not describe the convex hull of degree-sum-constrained simultaneous assignments in general.
Note that this remains the case even if one leaves out either constraints~(\ref{eq:simultan:degLp:deg})~or~(\ref{eq:simultan:degLp:card}).
The following definition and two theorems from~\cite[Page 594-598]{SchrijverYellowBook} will be useful.

\begin{definition}
  An integer matrix $M\in\Z^{m\times n}$ is \emph{bidirected} if $\sum_{i=1}^m|M_{ij}|=2$ for all $j\in\{1,\dots,n\}$.
\end{definition}

For a matrix $M\in\Z^{m\times n}$ and vectors $a,b\in\Z^m$ and $c,d\in\Z^n$, we consider the integer solutions of
\begin{equation}\label{eq:simultan:bidirProblem}
  \{x\in\Z^n : d\leq x\leq c, a\leq Mx\leq b\}.
\end{equation}

\begin{theorem}\label{thm:simultan:bidirPolytope}
  For a bidirected matrix $M\in\Z^{m\times n}$ and for arbitrary vectors $a,b\in\Z^m$ and $c,d\in\Z^n$, the convex hull of the integer solutions of~(\ref{eq:simultan:bidirProblem}) is described by the following system.
  \vspace{-3mm}
  \begin{subequations}
    \begin{align}\label{lp:simultan:bidirProglemHull}
      \tag{LP5}
      \quad\quad\quad\quad\quad&&\nonumber\\
      x&\in\R^n&\\
      d\leq x&\leq c&\\
      a\leq M&x\leq b&\\
      \frac{1}{2}((\chi_U-\chi_W)M+\chi_F-\chi_H)x&\leq \left\lfloor\frac{1}{2}(b(U)-a(W)+c(F)-d(H))\right\rfloor&\nonumber\\
                               &\hspace{-1.75cm}\text{for all disjoint } U,W\subseteq\{1,\dots,m\} \text{ and for all partition } F,H&\nonumber\\
                               &\hspace{-1.7cm}\text{of } \delta(U\cup W) \text{ with } b(U)-a(W)+c(F)-d(H)\text{ odd,}&\label{eq:simultan:bidirProglemHull:oddCon}
    \end{align}
  \end{subequations}
  where $\delta(U\cup W)=\{j\in \{1,\dots,n\} : \sum_{i\in U\cup W}|M_{ij}|=1\}$.
\end{theorem}
\begin{theorem}\label{thm:simultan:bidirSolvable}
  For a bidirected matrix $M\in\Z^{m\times n}$, and arbitrary vectors $a,b\in\Z^m$, $c,d\in\Z^n$ and $w\in\Q^n$, an integer vector $x$ maximizing $wx$ over~(\ref{eq:simultan:bidirProblem}) can be found in strongly polynomial time.
\end{theorem}
In what follows, we show that Theorems~\ref{thm:simultan:bidirPolytope}~and~\ref{thm:simultan:bidirSolvable} also hold in the slightly more general case when $M\in\Z^{m\times n}$ is such that $\sum_{i=1}^m|M_{ij}|\leq 2$ for all $j\in[n]$.

\begin{corollary}\label{cor:simultan:subBidirPolytope}
  Let $M\in\Z^{m\times n}$ be a matrix such that $\sum_{i=1}^m|M_{ij}|\leq 2$ for all ${j\in\{1,\dots,n\}}$, and let $a,b\in\Z^m$, $c,d\in\Z^n$.
  Then, the convex hull of the integer solutions of~(\ref{eq:simultan:bidirProblem}) is described by~(\ref{lp:simultan:bidirProglemHull}).
\end{corollary}
\begin{proof}
  Without loss of generality, one can assume that $M$ has no full-zero columns.
  The proof is by induction on the number of columns $j$ for which $\sum_{i=1}^m|M_{ij}|=1$.
  If there are no such columns, then the statement follows by Theorem~\ref{thm:simultan:bidirPolytope}.
  Let $j$ be the index of a column for which $\sum_{i=1}^m|M_{ij}|=1$.
  Now, one can simply add constraint $d_j\leq x_j\leq c_j$ explicitly to the system --- which is clearly redundant as it already appears in $d\leq x\leq c$.
  This means that $M$ is extended with an $(m+1)^{\text{th}}$ row $q=\chi_j\in\{0,1\}^n$, and hence by induction,~(\ref{lp:simultan:bidirProglemHull}) describes the convex hull of the extended problem.
  For the extended model, any constraints in~(\ref{eq:simultan:bidirProglemHull:oddCon}) with $U,W\subseteq\{1,\dots,m+1\}$ for which $m+1\in U\cup W$ is also implied by those generated before $q$ has been added to $M$.
  This completes the proof of the statement.
  \qed
\end{proof}

\begin{corollary}\label{cor:simultan:subBidirSolvable}
  Let $M\in\Z^{m\times n}$ be a matrix such that $\sum_{i=1}^m|M_{ij}|\leq 2$ for all $j\in\{1,\dots,n\}$, and let $d,c\in\Z^n,a,b\in\Z^m,w\in\Q^n$.
  Then, an integer vector $x$ maximizing $wx$ over~(\ref{eq:simultan:bidirProblem}) can be found in strongly polynomial time.
\end{corollary}
\begin{proof}
  Similarly to the proof of Corollary~\ref{cor:simultan:subBidirPolytope}, one can achieve that $\sum_{i=1}^m|M_{ij}|=2$ for all $j$, and hence Theorem~\ref{thm:simultan:bidirSolvable} can be applied.
  \qed
\end{proof}

Now, we show that the locally laminar simultaneous assignment problem can be formulated in such a way that it fits the framework given by~(\ref{eq:simultan:bidirProblem}), where $M$ is such that $\sum_{i=1}^m|M_{ij}|\leq 2$ for all $j$ --- and hence one can apply Corollaries~\ref{cor:simultan:subBidirPolytope}~and~\ref{cor:simultan:subBidirSolvable}.
First, consider the following notation.
For a laminar system $\mc F$, let $\mc C(\mc F)$ denote the inclusion-wise maximal sets in $\mc F$.
The maximal sets in $\mc F$ inside a member $X$ in $\mc F$ will be denoted by $\mc C(\mc F,X)$.

For a degree-sum-constrained simultaneous assignment $x$, let $y^v_F=\sum_{e\in F}x_e$ for $v\in V$ and $F\in\mc F_v$.
Furthermore, let $z_L=\sum_{v\in L}\sum_{e\in\Delta_G(v)}x_e$ for $L\in\mc L$.
That is, $y^v_F$ is the $x$-degree of node $v$ restricted to $F$ and $z_L$ is the sum of the $x$-degrees of the nodes in $L$ --- which appear as the left-hand side of~(\ref{eq:simultan:degLp:deg})~and~(\ref{eq:simultan:degLp:card}), respectively.
By definition,
\begin{equation}\label{eq:simultan:reqY}
  y^v_F=\sum_{e\in F\setminus\bigcup\mc C(\mc F_v,F)}x_e+\sum_{F'\in\mc C(\mc F_v,F)}y^v_{F'}
\end{equation}
holds for all $v\in V$ and $F\in\mc F_v$.
Without loss of generality, assume that $\{v\}\in\mc L$ for all $v\in V$ (if this is not the case for some $v\in V$, then one can add $\{v\}$ to $\mc L$ and set $g(\{v\})=\infty$).
Then,
\begin{equation}\label{eq:simultan:zForSingeltons}
  z_{\{v\}}=\sum_{e\in \Delta_G(v)\setminus\bigcup\mc C(\mc F_v)}x_e+\sum_{F\in\mc C(\mc F_v)}y^v_F
\end{equation}
holds for all $v\in V$ as well.
Similarly to~(\ref{eq:simultan:reqY}),
\begin{equation}\label{eq:simultan:reqZ}
  z_L=\sum_{v\in L\setminus\bigcup\mc C(\mc L,L)}z_{\{v\}}+\sum_{L'\in\mc C(\mc L,L)}z_{L'}
\end{equation}
holds for all $L\in\mc L$.
Considering $x,y$ and $z$ as variables, and combining (\ref{lp:simultan:degLp})~with equations~(\ref{eq:simultan:reqY}),~(\ref{eq:simultan:zForSingeltons})~and~(\ref{eq:simultan:reqZ}), one obtains the following linear program.
\begin{subequations}
  \begin{align}\label{lp:simultan:degLpExtension}
    \tag{LP2}
    \max&\sum_{st\in E}w_{st}x_{st}\\
    \mbox{s.t.}\quad\quad\quad\quad\quad&&\nonumber\\
    x&\in\R_+^{E}&\\
    y^v&\in\R_+^{\mc F_v}&\forall v\in V\\
    z&\in\R_+^{\mc L}&\\
    x&\leq c &\label{eq:simultan:degLpExtension:cap}\\
    y^v_{\Delta_H(v)}&\leq b_H(v)&\forall H\in\mc H,v\in V_H\label{eq:simultan:degLpExtension:deg}\\
    z&\leq g &\label{eq:simultan:degLpExtension:degSum}\\
    \sum_{e\in F\setminus\bigcup\mc C(\mc F_v,F)}\!\!x_e+\!\sum_{F'\in\mc C(\mc F_v,F)}y^v_{F'}-y^v_F&=0&\forall v\in V,\ \forall F\in\mc F_v\label{eq:simultan:degLpExtension:c1}\\
    \sum_{e\in \Delta_G(v)\setminus\bigcup\mc C(\mc F_v)}\!\!x_e+\!\sum_{F\in\mc C(\mc F_v)}y^v_F-z_{\{v\}}&=0&\forall v\in V\label{eq:simultan:degLpExtension:c2}\\
    \sum_{v\in L\setminus\bigcup\mc C(\mc L)}z_{\{v\}}+\sum_{L'\in\mc C(\mc L)}z_{L'}-z_L&=0&\forall L\in \mc L\setminus\{\{v\}:v\in V\}\label{eq:simultan:degLpExtension:c3}
  \end{align}
\end{subequations}

By construction, the solutions to~(\ref{lp:simultan:degLpExtension}) restricted to $x$ are exactly the feasible simultaneous assignments.
Note that each variable appears at most twice in constraints~(\ref{eq:simultan:degLpExtension:c1}),~(\ref{eq:simultan:degLpExtension:c2}) and~(\ref{eq:simultan:degLpExtension:c3}) with coefficient $1$ or $-1$, hence the matrix $M$ given by these three sets of constraints is such that $\sum_{i=1}^m|M_{ij}|\leq 2$.
Therefore, Corollary~\ref{cor:simultan:subBidirSolvable} immediately implies the following.
\begin{theorem}\label{thm:simultan:locLamCaseSolvable}
  If $\mc H$ is locally laminar, then the simultaneous assignment problem can be solved in strongly polynomial time.
\end{theorem}
As it has been already mentioned in Section~\ref{sec:simultan:intro}, the hierarchical $b$-matching problem is exactly the simultaneous assignment problem with $\mc H=\emptyset$.
Since in this case $\mc H$ is locally laminar, Theorem~\ref{thm:simultan:locLamCaseSolvable} can be applied.
\begin{corollary}
  The weighted, capacitated hierarchical $b$-matching problem is solvable in strongly polynomial time.
\end{corollary}
In the case when $\mc L=\emptyset$ and $\mc H$ is laminar, we get back the laminar matchoid problem, hence we obtain the following.
\begin{corollary}
  The weighted laminar matchoid problem can be solved in strongly polynomial time.
\end{corollary}

In fact, Theorem~\ref{thm:simultan:locLamCaseSolvable} applies even if we pose both \emph{lower} and upper bounds in constraints (\ref{eq:simultan:degLpExtension:cap}),~(\ref{eq:simultan:degLpExtension:deg}) and~(\ref{eq:simultan:degLpExtension:degSum}).
This immediately implies the following:
\begin{theorem}
  The locally laminar case can be also solved when both \emph{lower} and upper bounds are given in~(\ref{eq:simultan:degLp:cap}),~(\ref{eq:simultan:degLp:deg})~and~(\ref{eq:simultan:degLp:card}), that is, on the capacities of the edges, on the degrees in each subgraph $H\in\mc H$ and on the degree-sum in each $L\in\mc L$.
\end{theorem}

\begin{corollary}
  The weighted, capacitated hierarchical $b$-matching problem is solvable in strongly polynomial time even when both \emph{lower} and upper bounds are given on the capacities of the edges and on the degree sums in each member of the laminar family.
\end{corollary}

Note that Theorem~\ref{thm:simultan:bidirPolytope} gives a description of the convex hull of the integer points of~(\ref{lp:simultan:degLpExtension}).
Substituting variables $y$ and $z$, this in turn implies a description of the convex hull of the integer points of~(\ref{lp:simultan:degLp}) when $\mc H$ is locally laminar.
These constraints will be referred to as \emph{projected blossom inequalities}.

\begin{theorem}\label{thm:simultan:locallyLaminarPolytope}
  If $\mc H$ is locally laminar, then~(\ref{lp:simultan:degLp}) extended with the projected blossom inequalities describes the convex hull of simultaneous assignments.
\end{theorem}

The projected blossom inequalities can be used to strengthen~(\ref{lp:simultan:degLp})~if $\mc H$ is not locally laminar.
Namely, one can add the projected blossom inequalities to~(\ref{lp:simultan:degLp})~for all $F\subseteq E$ for which $\mc H$ becomes locally laminar when restricted to $F$:
\begin{subequations}
  \begin{align}\label{lp:simultan:degLpStrong}
    \tag{LP1$^*$}
    \max\sum_{st\in E}&w_{st}x_{st}\\
    \mbox{s.t.}\quad\quad\quad\quad\quad&&\nonumber\\
    (\ref{eq:simultan:degLp:nonNeg})~(\ref{eq:simultan:degLp:cap})~(\ref{eq:simultan:degLp:deg})~(\ref{eq:simultan:degLp:card})& & \nonumber\\
    \mathcal B(F) & &\forall F\subseteq E:\text{$\restr{\mc H}{F}$ is locally laminar},
  \end{align}
\end{subequations}
where $\mathcal B(F)$ denotes the set of projected blossom inequalities when the problem is restricted to $F\subseteq E$.
By Theorem~\ref{thm:simultan:locallyLaminarPolytope}, (\ref{lp:simultan:degLpStrong}) becomes integer when the problem is restricted to a locally laminar edge set.
The integrality gap of this strengthened linear program will be investigated in Section~\ref{sec:simultan:apx}.

\subsection{When the Graph is a Tree}\label{sec:simultan:trees}
In Section~\ref{sec:simultan:hardness}, we will see that the simultaneous assignment problem is hard even if $G$ consists of node-disjoint claws and the size of $\mc H$ is two.
However, assuming that $G$ is a tree and $\mc L=\emptyset$, the problem becomes solvable provided that a natural assumption on $\mc H$ holds.
Motivated by the first application described in the introduction, consider the following definition.

\begin{definition}
  We say that $\mc H=\{H_1,\dots,H_k\}$ has the \emph{local-interval property} if, for each node $v\in V$, there exists a permutation $H_{i_1},\dots,H_{i_k}$ under which each edge in $\Delta_v$ is included in a (possibly empty) interval $H_{i_p},\dots,H_{i_q}$.
\end{definition}

Note that the local-interval property holds in the first motivation mentioned in the introduction when each customer selects an interval of the events.
Under these assumptions, the problem becomes solvable:

\begin{theorem}\label{thm:simultan:treeInterval}
  Let $G=(V,E)$ be a tree, let $\mc L=\emptyset$ and assume that $\mc H$ has the local-interval property.
  Then, the simultaneous assignment polyhedron is described by~(\ref{lp:simultan:degLp}) and hence the problem can be solved in strongly polynomial time.
\end{theorem}
\begin{proof}
  It suffices to show that the matrix of~(\ref{lp:simultan:degLp}) is a network matrix~\cite[Page 151]{AF11}.
  First, we introduce the following notations.
  For each $v\in V$, there exists a permutation $H_{i_1^v},\dots,H_{i_k^v}$ under which each edge $e\in\Delta(v)$ is contained in an interval by the local interval property.
  Let $p^e_v=\argmin\{j : e\in H_{i_j^v}\}$ and $q^e_v=\argmax\{j : e\in H_{i_j^v}\}$ for $v\in V$ and $e\in\Delta(v)$.

  The construction of the directed tree $T$ corresponding to the network matrix is as follows.
  For each $v\in V$, add $(k+1)$ new nodes $z_{0}^v,\dots,z_{k}^v$ to $T$.
  Also add a new arc from $z_{j-1}^v$ to $z_{j}^v$ for each $j\in\{1,\dots,k\}$, which is associated with $H_{i_j^v}$ --- and with the corresponding row of the network matrix.
  Let $r\in V$ be arbitrary and appoint it as the root node of $G$.
  For each edge $e=uv\in E$ of $G$ such that $u$ is the parent node of $v$ with respect to root $r$, unify $z^u_{q^e_u}$ and $z^v_{p^e_v}$ in $T$.
  This way, $T$ becomes a directed tree.
  Now, we are ready to describe the non-tree arcs defining the network matrix.
  For each edge $e=uv\in E$ of $G$ such that $u$ is the parent node of $v$ with respect to root $r$, add a new non-tree arc from $z^v_{q^e_v}$ to $z^u_{p^e_u}$.
  By construction, the base cycle of this arc contains exactly those arcs of $T$ (as forward arcs) which are associated with those rows of the matrix of~(\ref{lp:simultan:degLp}) in which the entry corresponding to $uv$ is one.
  \qed
\end{proof}

Later, we will give examples showing that one cannot drop any of the assumptions in Theorem~\ref{thm:simultan:treeInterval} (see Figures~\ref{fig:simultan:gapExamplek2}~and~\ref{fig:simultan:pseodoTreeGapEg}).
Observe that $\mc H$ has the local-interval property if its size is two, therefore, one obtains the following corollary of Theorem~\ref{thm:simultan:treeInterval}.

\begin{corollary}\label{cor:simultan:tree2Subgr}
  Let $G$ be a tree, let $\mc L=\emptyset$ and assume that the size of $\mc H$ is two.
  Then, the simultaneous assignment polyhedron is described by~(\ref{lp:simultan:degLp}) and hence the problem can be solved in strongly polynomial time.
\end{corollary}

\section{Hardness Results}\label{sec:simultan:hardness}
As we have already seen in the introduction, it is NP-complete to decide whether there exists a simultaneous assignment satisfying constraints~(\ref{eq:simultan:degLp:deg}) with equality.
In this section, multiple special cases of the problem will be shown hard to approximate.
All presented hardness results hold true even for bipartite graphs in the uncapacitated case.

\subsection{The Case of Two Subgraphs}
In this section, we show that the unweighted simultaneous assignment problem cannot be approximated arbitrarily well in multiple special cases.

\begin{theorem}\label{thm:simultan:sapAPXhard}
  The unweighted simultaneous assignment problem is NP-hard to approximate within any factor smaller than $\frac{570}{569}$ even if $|\mc H|=2$, the connected components of $G$ are claws and the size of every member in $\mc L$ is two.
\end{theorem}
\begin{proof}
  Given are three finite disjoint sets $X,Y,Z$ and a set of hyperedges $\mc E\subseteq X\times Y\times Z$, a subset of the hyperedges $F\subseteq\mc E$ is called \emph{$3$-dimensional matching} if $x_1\neq x_2, y_1\neq y_2$ and $z_1\neq z_2$ for any two distinct triples $(x_1, y_1, z_1), (x_2, y_2, z_2) \in F$.
  Finding a maximum-size 3-dimensional matching $F\subseteq\mc E$ cannot be approximated arbitrarily unless P=NP.
  In fact, the problem remains NP-hard to approximate better than $\frac{95}{94}$ even for \emph{$2$-regular} instances, that is, when each element of $X\cup Y\cup Z$ occurs in exactly two triples in $\mc E$~\cite{CHLEBIK2006320}.
  To reduce the 2-regular 3-dimensional matching problem to the simultaneous assignment problem, consider the following construction.

  Let $e^z_1$ and $e^z_2$ denote the two hyperedges incident to $z$ for all $z\in Z$.
  As each element occurs in exactly two triples, $|X|=|Y|=|Z|$ holds.
  First, define a bipartite graph $G=(S,T;E)$, where $S=X\cup \mc E\cup Y$, $T=\mc E$ and $E$ is as follows.
  For each $s\in S\cap(X\cup Y)$, add an edge between $s$ and the two hyperedges $e\in T$ incident to $s$; and connect the two occurrences $e\in S\cap\mc E$ and $e\in T\cap\mc E$ of each hyperedge $e\in\mc E$.
  Let $H_1$ and $H_2$ consist of the edges incident to $S\setminus Y$ and $S\setminus X$, respectively.
  We set $b_{H_1}(s)=b_{H_2}(s)=1$ for each $s\in S$.
  Let $L_z^T=\{e^z_1,e^z_2\}\subseteq T$, $\mc L=\{L_z^T : z\in Z\}$ and $g\equiv 3$.
  At the end of this proof, we will modify this construction so that the bipartite graph defined here becomes the union of disjoint claws.

  Figures~\ref{fig:simultan:3dimMatching}~and~\ref{fig:simultan:3dimMatchingConstruction} show an instance of the 2-regular 3-dimensional matching problem and the corresponding instance of the simultaneous assignment as per the construction above, respectively.
  Each hyperedge is represented by a unique line style,
  for example, the dotted lines in Figures~\ref{fig:simultan:3dimMatching}~and~\ref{fig:simultan:3dimMatchingConstruction} represent the same hyperedge $e=(x_1,y_1,z_1)$.
  Note that the edges represented by a straight line in Figure~\ref{fig:simultan:3dimMatchingConstruction} do not represent hyperedges, but the edges between hyperedges.

  Suppose that for some $\alpha\geq 1$, there exists a polynomial-time $\alpha$-approximation algorithm for the unweighted simultaneous assignment problem, and let $M$ be an $\alpha$-approximate solution.
  We construct a feasible solution $F$ to the 3-dimensional matching problem from $M$ in polynomial time and investigate its approximation ratio.
  For each node $e^z_i\in T\cap\mc E$, if there exist two edges $xe^z_i,ye^z_i\in M$ incident to $e^z_i\in T\cap\mc E$, then add $(x,y,z)$ to $F$.
  Note that $F$ is a feasible 3-dimensional matching because every element of $X\cup Y$ appears in at most one of the selected tuples (since $b_{H_1}(s)=b_{H_2}(s)=1$ for all $s\in S$) and the degree of at most one of $e^z_1\in T\cap\mc E$ and $e^z_2\in T\cap\mc E$ can be two (since $\{e^z_1,e^z_2\}\in\mc L$ and $g\equiv 3$).

  In what follows, the approximation ratio $\frac{|F^*|}{|F|}$ will be investigated, where $F^*$ denotes a largest 3-dimensional matching.
  Observe that $|F^*|\geq\frac{|\mc E|}{4}=\frac{|Z|}{2}$, because any inclusion-wise maximal 3-dimensional matching consists of at least $\frac{|\mc E|}{4}$ hyperedges.
  Let $M^*$ be an optimal solution to the simultaneous assignment problem.
  Note that $|M^*|\leq 3|Z|$, since the number of edges of $M$ that are incident to $L_z^T$ is at most 3 for $z\in Z$.
  Hence, one gets that
  \begin{equation}\label{eq:simultan:apxProofBound1}
    |M^*|\leq 3|Z|\leq 6|F^*|.
  \end{equation}

  By construction, it is clear that $|M^*|=|\mc E|+|F^*|$ and $|M|=|\mc E|+|F^*|$ hold, therefore
  \begin{equation}\label{eq:simultan:apxProofBound2}
    |M^*|-|M|=|F^*|-|F|
  \end{equation}
  follows.
  Using these observations, one gets that
  \begin{multline}\label{eq:simultan:apxProofBound3}
    \frac{|F|}{|F^*|}=\frac{|F^*|-(|M^*|-|M|)}{|F^*|}=1-\frac{|M^*|-|M|}{|F^*|}\geq 1-6\frac{|M^*|-|M|}{|M^*|}\\
    =1-6(1-\frac{|M|}{|M^*|})=1-6(1-\frac{1}{\alpha})=\frac{6}{\alpha}-5,
  \end{multline}
  where the first equation holds by (\ref{eq:simultan:apxProofBound2}) and the inequality follows by (\ref{eq:simultan:apxProofBound1}).
  Hence one gets that $\frac{|F^*|}{|F|}\leq\frac{1}{6/\alpha-5}$, provided that $|F|\neq 0$ and $\alpha<\frac{6}{5}$.
  To complete the proof, note that the right-hand side, $\frac{1}{6/\alpha-5}$, tends to $1$ as $\alpha$ monotone decreasingly tends to $1$, hence $\frac{|F^*|}{|F|}$ tends to 1 as well.
  This means that if one has an $\alpha$-approximation algorithm for the unweighted simultaneous assignment problem for some $\alpha<\frac{6}{5}$, then it is possible to construct a $(\frac{1}{6/\alpha-5})$-approximate solution to the 2-regular 3-dimensional matching problem in polynomial time.
  The latter problem is NP-hard to $\beta$-approximate for any $\beta<\frac{95}{94}$ even if each element of $X\cup Y\cup Z$ occurs in exactly two triples in $\mc E$, see~\cite{CHLEBIK2006320}, which implies that the unweighted simultaneous assignment problem is hard to $\alpha$-approximate whenever $\frac{1}{6/\alpha-5}<\frac{95}{94}$.
  Reordering this inequality, one gets that the unweighted simultaneous assignment problem is NP-hard to $\alpha$-approximate if $\alpha<\frac{570}{569}$.

  To complete the proof of the theorem, we modify the construction of the simultaneous assignment problem instance such that $G$ becomes the disjoint union of claws.
  In the construction above, the degree of each node of $s\in S\cap(X\cup Y)$ is exactly two and the degree-bound constraints posed in the two subgraphs $H_1$ and $H_2$ is $b_{H_1}(s)=b_{H_2}(s)=1$.
  For every such node $s\in S\cap (X\cup Y)$, add a disjoint copy of $s$ to $S$ and move one of the edges incident to $s$ to this new node.
  This way, the modified graph is the node-disjoint union of claws.
  Let $L_{s}^S$ consist of $s$ and its newly added copy, and add $L_{s}^S$ to $\mc L$.
  By setting $g(L_{s}^S)=1$ for $s\in S\cap (X\cup Y)$, one can ensure that any feasible simultaneous assignment problem contains at most one edge which is incident to $L_{s}^S$.
  Figure~\ref{fig:simultan:3dimMatchingConstructionTree} illustrates the construction when applied to the instance shown in Figure~\ref{fig:simultan:dmEg2}.
  Clearly, there is one-to-one correspondence between the feasible solutions to the original and the modified problems such that the cardinalities of the solutions are the same, which completes the proof of the theorem.
  \qed
\end{proof}

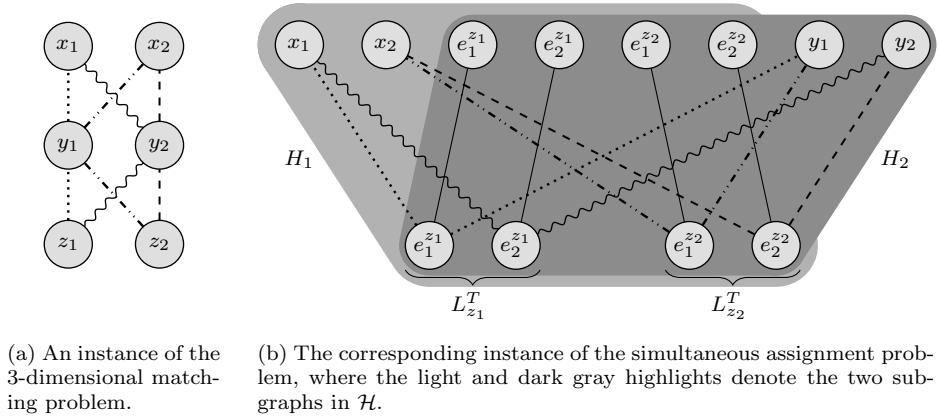
\begin{figure}
  \begin{subfigure}[t]{.23\textwidth}
    \centering
    \begin{tikzpicture}[xscale=.8,yscale=.75]
      \begin{scope}[yshift=.8cm]
        \SetVertexMath
        \grEmptyPath[form=1,x=0,y=0,RA=1.5,rotation=0,prefix=z]{2}
        \grEmptyPath[form=1,x=0,y=1.75,RA=1.5,rotation=0,prefix=y]{2}
        \grEmptyPath[form=1,x=0,y=3.5,RA=1.5,rotation=0,prefix=x]{2}
        \node[below=.57cm of z1] (a) {};
        \draw[dotted,line width=.9] (z0) -- (y0);
        \draw[dotted,line width=.9] (y0) -- (x0);

        \draw[dashed,line width=.75] (z1) -- (y1);
        \draw[dashed,line width=.75] (y1) -- (x1);

        \draw[wavy] (z0) -- (y1);
        \draw[wavy] (y1) -- (x0);

        \draw[dashdotdotted,line width=.9] (z1) -- (y0);
        \draw[dashdotdotted,line width=.9] (y0) -- (x1);
      \end{scope}
    \end{tikzpicture}
    \caption{An instance of the 3-dimensional matching problem.}\label{fig:simultan:3dimMatching}
  \end{subfigure}
  \hfill
  \begin{subfigure}[t]{.73\textwidth}
    \centering
    \begin{tikzpicture}[scale=.95,xscale=1,yscale=.8]
      \SetVertexMath
      \tikzset{VertexStyle/.append style = {minimum size = 18pt,inner sep=0pt}}
      \Vertex[x=0,y=3.5,L=x_1]{x0}
      \Vertex[x=1.2,y=3.5,L=x_2]{x1}

      \Vertex[x=2.4,y=3.5,L=e_1^{z_1}]{se1}
      \Vertex[x=3.6,y=3.5,L=e_2^{z_1}]{se2}
      \Vertex[x=4.8,y=3.5,L=e_1^{z_2}]{se3}
      \Vertex[x=6,y=3.5,L=e_2^{z_2}]{se4}

      \Vertex[x=7.2,y=3.5,L=y_1]{y0}
      \Vertex[x=7.2+1.2,y=3.5,L=y_2]{y1}

      \Vertex[x=1.8,y=0,L=e_1^{z_1}]{e0}
      \Vertex[x=3,y=0,L=e_2^{z_1}]{e1}
      \Vertex[x=5.4,y=0,L=e_1^{z_2}]{e3}
      \Vertex[x=6.6,y=0,L=e_2^{z_2}]{e2}

      \draw[dotted,line width=.9] (e0) -- (x0);
      \draw[dotted,line width=.9] (e0) -- (y0);

      \draw[wavy] (e1) -- (x0);
      \draw[wavy] (e1) -- (y1);

      \draw[dashed,line width=.75] (e2) -- (x1);
      \draw[dashed,line width=.75] (e2) -- (y1);

      \draw[dashdotdotted,line width=.9] (e3) -- (y0);
      \draw[dashdotdotted,line width=.9] (e3) -- (x1);

      \draw[] (e1) to [bend left = 0] (se2);
      \draw[] (e0) -- (se1);

      \draw[] (e2) to [bend left = 0] (se4);
      \draw[] (e3) -- (se3);

      \begin{scope}[decoration={brace,mirror,amplitude=2mm,raise=3.25mm}]
        \begin{scope}[every node/.style={midway,left,yshift=-8mm,xshift=3mm}]
          \draw[decorate] (e0.west) -- (e1.east)node(C1){$L^T_{z_1}$};
          \draw[decorate] (e3.west) -- (e2.east)node(C2){$L^T_{z_2}$};
        \end{scope}
      \end{scope}

      \begin{pgfonlayer}{background}
        \begin{scope}[transparency group,opacity=.8]
          \highlight{11mm}{black!30, fill=black!30}{(x0.center) to [bend left = 0] (se4.center)  to (e2.center) to (e0.center) to (x0.center)};
        \end{scope}
        \begin{scope}[transparency group,opacity=.8]
          \highlight{8mm}{black!45, fill=black!45}{(e2.center) to (y1.center) to (se1.center) to (e0.center) to (e2.center)};
        \end{scope}
        \begin{pgfonlayer}{background}
        \end{pgfonlayer}
      \end{pgfonlayer}
      \node at (8.25,1.5){$H_2$};
      \node at (0,1.5){$H_1$};
    \end{tikzpicture}
    \caption{The corresponding instance of the simultaneous assignment problem, where the light and dark gray highlights denote the two subgraphs in $\mc H$.}\label{fig:simultan:3dimMatchingConstruction}
  \end{subfigure}
  \caption{\vspace{-5mm}Illustration of the proof of Theorem~\ref{thm:simultan:sapAPXhard}.
  }\label{fig:simultan:dmEg2}
\end{figure}

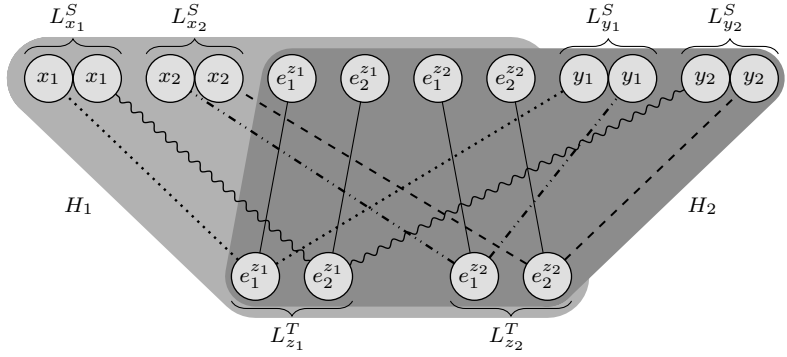
\begin{figure}
  \centering
  \begin{tikzpicture}[xscale=.8,yscale=.75]
    \SetVertexMath
    \tikzset{VertexStyle/.append style = {minimum size = 18pt,inner sep=0pt}}
    \Vertex[x=0,y=3.5,L=x_1]{x0}
    \Vertex[x=0+.8,y=3.5,L=x_1]{x0_}
    \Vertex[x=1.2+.8,y=3.5,L=x_2]{x1}
    \Vertex[x=1.2+.8+.8,y=3.5,L=x_2]{x1_}

    \Vertex[x=2.4+.8+.8,y=3.5,L=e_1^{z_1}]{se1}
    \Vertex[x=3.6+.8+.8,y=3.5,L=e_2^{z_1}]{se2}
    \Vertex[x=4.8+.8+.8,y=3.5,L=e_1^{z_2}]{se3}
    \Vertex[x=6+.8+.8,y=3.5,L=e_2^{z_2}]{se4}

    \Vertex[x=7.2+.8+.8,y=3.5,L=y_1]{y0}
    \Vertex[x=7.2+.8+.8+.8,y=3.5,L=y_1]{y0_}
    \Vertex[x=7.2+1.2+.8+.8+.8,y=3.5,L=y_2]{y1}
    \Vertex[x=7.2+1.2+.8+2*.8+.8,y=3.5,L=y_2]{y1_}

    \Vertex[x=1.8+.8+.8,y=0,L=e_1^{z_1}]{e0}
    \Vertex[x=3+.8+.8,y=0,L=e_2^{z_1}]{e1}
    \Vertex[x=5.4+.8+.8,y=0,L=e_1^{z_2}]{e3}
    \Vertex[x=6.6+.8+.8,y=0,L=e_2^{z_2}]{e2}

    \draw[dotted,line width=.9] (e0) -- (x0);
    \draw[dotted,line width=.9] (e0) -- (y0);

    \draw[wavy] (e1) -- (x0_);
    \draw[wavy] (e1) -- (y1);

    \draw[dashed,line width=.75] (e2) -- (x1_);
    \draw[dashed,line width=.75] (e2) -- (y1_);

    \draw[dashdotdotted,line width=.9] (e3) -- (y0_);
    \draw[dashdotdotted,line width=.9] (e3) -- (x1);

    \draw[] (e1) to [bend left = 0] (se2);
    \draw[] (e0) -- (se1);

    \draw[] (e2) to [bend left = 0] (se4);
    \draw[] (e3) -- (se3);

    \begin{scope}[decoration={brace,mirror,amplitude=2mm,raise=3.25mm}]
      \begin{scope}[every node/.style={midway,left,yshift=-8mm,xshift=3mm}]
        \draw[decorate] (e0.west) -- (e1.east)node(C1){$L^T_{z_1}$};
        \draw[decorate] (e3.west) -- (e2.east)node(C2){$L^T_{z_2}$};
      \end{scope}
    \end{scope}

    \begin{scope}[decoration={brace,amplitude=2mm,raise=3.25mm}]
      \begin{scope}[every node/.style={midway,left,yshift=8mm,xshift=3mm}]
        \draw[decorate] (x0.west) -- (x0_.east)node(C1){$L_{x_1}^S$};
        \draw[decorate] (x1.west) -- (x1_.east)node(C1){$L_{x_2}^S$};
        \draw[decorate] (y0.west) -- (y0_.east)node(C1){$L_{y_1}^S$};
        \draw[decorate] (y1.west) -- (y1_.east)node(C1){$L_{y_2}^S$};
      \end{scope}
    \end{scope}

    \begin{pgfonlayer}{background}
      \begin{scope}[opacity=.8,transparency group]
        \highlight{11mm}{black!30, fill=black!30}{(x0.center) to [bend left = 0] (se4.center)  to (e2.center) to (e0.center) to (x0.center)}
      \end{scope}
      \begin{scope}[opacity=.8,transparency group]
        \highlight{8mm}{black!45, fill=black!45}{(e2.center) to (y1_.center) to (se1.center) to (e0.center) to (e2.center)}
      \end{scope}
      \begin{pgfonlayer}{background}
      \end{pgfonlayer}
    \end{pgfonlayer}
    \node at (10.75,1.25){$H_2$};
    \node at (.5,1.25){$H_1$};
  \end{tikzpicture}
  \caption{Illustration of the splitting operation used in the proof of Theorem~\ref{thm:simultan:sapAPXhard} to the instance given in Figure~\ref{fig:simultan:dmEg2}.}\label{fig:simultan:3dimMatchingConstructionTree}
\end{figure}

Next, we show that the problem is also hard to approximate when the size of each subgraph in $\mc H$ is assumed to be two.
To this end, we modify the construction given in the proof of Theorem~\ref{thm:simultan:sapAPXhard} as follows.
Let $G'=G$, $\mc L'=\mc L$, $g'=g$, and let $\mc H'$ consist of all subgraphs which are obtained as the intersection of a claw in $G$ and either $H_1$ or $H_2$.
Setting $b'_{H'}\equiv 1$ for all $H'\in\mc H'$, an equivalent problem is obtained in which each subgraph in $\mc H'$ consists of only two edges.
Hence, Theorem~\ref{thm:simultan:sapAPXhard} implies the following.

\begin{corollary}
  The unweighted simultaneous assignment problem is NP-hard to approximate within any factor smaller than $\frac{570}{569}$ even if all subgraphs in $\mc H$ consist of at most two edges, the size of all members of $\mc L$ is two and the connected components of $G$ are claws.
\end{corollary}

The proof of Theorem~\ref{thm:simultan:sapAPXhard} heavily relies on the cardinality constraints~(\ref{eq:simultan:degLp:card}).
Now, we assume that $\mc L=\emptyset$ and show that the weighted simultaneous assignment problem remains hard to approximate for two subgraphs.

\begin{theorem}\label{thm:simultan:sapAPXhard2}
  The weighted simultaneous assignment problem is NP-hard to approximate within any factor smaller than $\frac{760}{759}$, even if $|\mc H|=2$, $\mc L=\emptyset$, each occurring weight is either $1$ or $2$ and $G$ is bipartite with maximum degree at most four.
\end{theorem}
\begin{proof}
  Similarly to the proof of Theorem~\ref{thm:simultan:sapAPXhard}, let $X,Y,Z$ and $\mc E\subseteq X\times Y\times Z$ be an instance of the 2-regular 3-dimensional matching problem.
  We reduce the 2-regular 3-dimensional matching problem to the simultaneous assignment problem as follows.

  First, define a bipartite graph $G=(S,T;E)$, where $S=X\cup \mc E\cup Y$, $T=\mc E$ and $E$ is defined as follows.
  For each $s\in S\cap(X\cup Y)$, add an edge of weight 1 between $s$ and both hyperedges $e\in T$ incident to $s$; and connect each $e^z_2\in S\cap\mc E$ to hyperedges $e^z_1,e^z_2\in T$ with an edge of weight 2 for each $z\in Z$, where $e^z_1$ and $e^z_2$ denote the two hyperedges incident to $z$.
  Finally, also add an edge of weight 1 between $e_1^z\in S\cap\mc E$ and $e_1^z\in T\cap\mc E$ for each $z\in Z$.
  Let $S_1=S\setminus Y$ and $S_2=S\setminus X$, and define $H_i$ as the subgraph induced by $S_i\cup T$ for $i\in\{1,2\}$.
  We set $\mc L = \emptyset$ and $b_1(v)=b_2(v)=1$ for all $s\in S\cup T$.
  Figure~\ref{fig:simultan:3dimmatchingconstructionWAPX} shows the construction for the 2-regular 3-dimensional matching instance given by Figure~\ref{fig:simultan:3dimMatching}.
  Each hyperedge is represented by a unique line style, for example, the dotted lines in Figures~\ref{fig:simultan:3dimMatching}~and~\ref{fig:simultan:3dimmatchingconstructionWAPX} represent the same hyperedge $e=(x_1,y_1,z_1)$.
  Note that the edges represented by a straight line in Figure~\ref{fig:simultan:3dimmatchingconstructionWAPX} do not represent hyperedges, but the edges between hyperedges.

  Suppose that for some $\alpha\geq 1$, there exists a polynomial-time $\alpha$-approximation algorithm for the maximum-weight simultaneous assignment problem and let $M$ be an $\alpha$-approximate solution.
  For any $z\in Z$, consider the following three transformations:
  1) If both $e^z_1\in T$ and $e^z_2\in T$ have degree two in $M$, then replace the edges incident to $e^z_2\in T$ with edge $e^z_2e^z_2$.
  2) If exactly one of $e^z_1,e^z_2\in T$ has degree two in $M$, then connect the other one to $e^z_2\in S$ and remove all of its incident edges.
  3) If none of $e^z_1,e^z_2\in T$ has degree two in $M$, then add edges $e^z_1e^z_1$ and $e^z_2e^z_2\in T$, and remove any other incident edges.
  After these operations, $M$ remains feasible and its weight does not decrease, hence it remains $\alpha$-approximate.
  Let us construct a 3-dimensional matching $F$ from $M$ as follows.
  For each node $e^z_i\in T\cap\mc E$, if two edges $xe^z_i,ye^z_i\in M$ are incident to $e^z_i$ in $M$, then add $(x,y,z)$ to $F$.
  Note that $F$ is a feasible 3-dimensional matching, as at most one of $s^z_1,s^z_2$ has degree two in $M$ after the transformations.

  In what follows, the approximation ratio $\frac{|F^*|}{|F|}$ will be investigated, where $F^*$ denotes a largest 3-dimensional matching.
  Observe that $|F^*|\geq\frac{|\mc E|}{4}$, because any maximal 3-dimensional matching consists of at least $\frac{|\mc E|}{4}$ hyperedges.
  Let $M^*$ be an optimal simultaneous assignment and perform the transformations above on $M^*$.
  Note that $w(M^*)\leq 2|\mc E|$, since the sum of the weights of those edges of $M$ that are incident to a node $t\in T$ is at most two.
  Hence, we get that
  \begin{equation*}
    w(M^*)\leq 2|\mc E|\leq 8|F^*|.
  \end{equation*}

  By transformations 1), 2) and 3), both $w(M^*)=2|\mc E|-|Z|+|F^*|=3|Z|+|F^*|$ and $w(M)=2|\mc E|-|Z|+|F|=3|Z|+|F|$ hold, therefore
  \begin{equation*}
    w(M^*)-w(M)=|F^*|-|F|
  \end{equation*}
  follows.
  Similarly to~(\ref{eq:simultan:apxProofBound3}), this implies that $\frac{|F^*|}{|F|}\leq\frac{1}{8/\alpha-7}$ provided that $|F|\neq 0$ and $\alpha<\frac{8}{7}$.
  To complete the proof, note that the right-hand side, $\frac{1}{8/\alpha-7}$, tends to $1$ as $\alpha$ monotone decreasingly tends to 1, hence $\frac{|F^*|}{|F|}$ tends to $1$ as well.
  This means that if one has an $\alpha$-approximation algorithm for the weighted simultaneous assignment problem for some $\alpha<\frac{8}{7}$, then it is possible to construct a $(\frac{1}{8/\alpha-7})$-approximate solution to the 2-regular 3-dimensional matching problem in polynomial time.
  The latter problem is NP-hard to $\beta$-approximate for any $\beta<\frac{95}{94}$ even if each element of $X\cup Y\cup Z$ occurs in exactly two triples in $\mc E$~\cite{CHLEBIK2006320}, which implies that the simultaneous assignment problem is hard to $\alpha$-approximate whenever $(\frac{1}{8/\alpha-7})<\frac{95}{94}$.
  Reordering this inequality, one gets that the simultaneous assignment problem is NP-hard to $\alpha$-approximate if $\alpha<\frac{760}{759}$, which completes the proof.
  \qed
\end{proof}

\begin{figure}
  \centering
  \begin{tikzpicture}[xscale=.8,yscale=.75]
    \SetVertexMath
    \tikzset{VertexStyle/.append style = {minimum size = 18pt,inner sep=0pt}}
    \grEmptyPath[form=1,x=0,y=3.5,RA=1.2,rotation=0,prefix=x]{2}
    \Vertex[x=2.4,y=3.5,L=e_1^{z_1}]{se1}
    \Vertex[x=3.6,y=3.5,L=e_2^{z_1}]{se2}
    \Vertex[x=4.8,y=3.5,L=e_1^{z_2}]{se3}
    \Vertex[x=6,y=3.5,L=e_2^{z_2}]{se4}
    \grEmptyPath[form=1,x=7.2,y=3.5,RA=1.2,rotation=0,prefix=y]{2}

    \Vertex[x=1.8,y=0,L=e_1^{z_1}]{e0}
    \Vertex[x=3,y=0,L=e_2^{z_1}]{e1}
    \Vertex[x=5.4,y=0,L=e_1^{z_2}]{e3}
    \Vertex[x=6.6,y=0,L=e_2^{z_2}]{e2}

    \draw[dotted,line width=.9] (e0) -- (x0);
    \draw[dotted,line width=.9] (e0) -- (y0);

    \draw[wavy] (e1) -- (x0);
    \draw[wavy] (e1) -- (y1);

    \draw[dashed,line width=.75] (e2) -- (x1);
    \draw[dashed,line width=.75] (e2) -- (y1);

    \draw[dashdotdotted,line width=.9] (e3) -- (y0);
    \draw[dashdotdotted,line width=.9] (e3) -- (x1);

    \draw[] (e0) -- (se2) node[midway,xshift=-1.1mm,yshift=.5mm] {$2$};
    \draw[] (e1) -- (se2) node[midway,xshift=-1.7mm,yshift=-3mm] {$2$};
    \draw[] (e0) -- (se1);
    \draw[] (e2) -- (se4) node[midway,xshift=2mm,yshift=-5mm] {$2$};
    \draw[] (e3) -- (se4) node[midway,xshift=.25mm,yshift=-4mm] {$2$};
    \draw[] (e3) -- (se3);

    \begin{pgfonlayer}{background}
      \begin{scope}[opacity=.8,transparency group]
        \highlight{11mm}{black!30, fill=black!30}{(x0.center) to [bend left = 0] (se4.center)  to (e2.center) to (e0.center) to (x0.center)}
      \end{scope}
      \begin{scope}[opacity=.8,transparency group]
        \highlight{8mm}{black!45, fill=black!45}{(e2.center) to (y1.center) to (se1.center) to (e0.center) to (e2.center)}
      \end{scope}
      \begin{pgfonlayer}{background}
      \end{pgfonlayer}
    \end{pgfonlayer}

    \node at (8.25,1.25){$H_2$};
    \node at (0,1.25){$H_1$};

  \end{tikzpicture}
  \caption{Illustration of the proof of Theorem~\ref{thm:simultan:sapAPXhard2}. The unlabeled edges are of weight 1.}\label{fig:simultan:3dimmatchingconstructionWAPX}
\end{figure}
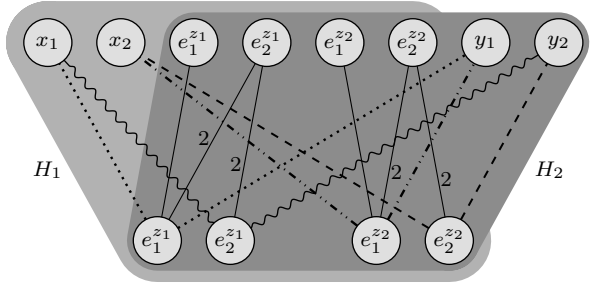

\subsection{The Case of Multiple Subgraphs}
It is quite natural to ask if there exists an $\alpha$-approximation algorithm for some constant $\alpha$ independent of the size of $\mc H$.
The next theorem shows that no such algorithm is possible, in fact, the approximation factor must grow asymptotically (essentially) linearly with the size of $\mc H$ unless P=NP.
\begin{theorem}
  The simultaneous assignment problem is $\Omega(|\mc H|^{1-\epsilon})$-inapproximable for all $\epsilon>0$ unless P=NP.
  The result holds even if $\mc L=\emptyset$.
\end{theorem}
\begin{proof}
  The \emph{maximum independent set problem} in a graph $G'=(V',E')$ is $\Omega(n^{1-\epsilon})$-inapprox\-imab\-le for all $\epsilon>0$~\cite{cliqueApxVeryHard}, where $n$ is the number of nodes in $V'$.
  We give an approximation-ratio preserving reduction from the maximum independent set problem to the simultaneous assignment problem.
  Let $G=(S,T;E)$ be a bipartite graph, where $S=V'$, $T=\{t\}$ and $E$ consists of all edges between $S$ and $T$.
  Let $\mc H = \{H_u : u\in V'\}$, where $H_u$ is the subgraph of $G$ induced by nodes $N_{G'}(u)\cup\{u,t\}$.
  Let $b_H\equiv 1$ for all $H\in\mc H$.
  For any subset $X\subseteq S$, $X$ can be covered by a simultaneous assignment in $G$ if and only if $X$ is independent in $G'$.
  To see this, observe that for any $u\in S$, if $u$ is covered by a simultaneous assignment, then no neighbor of $u$ may be covered because $b_{H_u}(t)=1$, hence the covered nodes form an independent set in $G'$.
  For the reverse direction, note that an independent set $X$ is covered by the simultaneous assignment $\chi_{\{st : s\in X\}}$.
  Since the size of $\mc H$ in the construction is exactly the number of nodes in $G'$, the simultaneous assignment problem is $\Omega(|\mc H|^{1-\epsilon})$-inapproximable for all $\epsilon>0$.
  \qed
\end{proof}

\subsection{Consequences to the $d$-Distance Matching Problem}\label{sec:simultan:hardness:conseqToDM}
In this section, we show that the hardness results above imply the inapproximability of both the weighted $d$-distance matching problem and the unweighted cyclic $d$-distance matching problem.
First, we prove that the weighted $d$-distance matching problem is hard to approximate.

\begin{theorem}\label{thm:simultan:wdmapxc}
  The weighted $d$-distance matching problem is NP-hard to approximate within any factor smaller than $\frac{760}{759}$, even if the maximum degree of the graph is at most $4$ and each occurring weight is either $1$ or $2$.
\end{theorem}
\begin{proof}
  Let us given an instance of the simultaneous assignment problem as per the statement of Theorem~\ref{thm:simultan:sapAPXhard2}, that is, let $G$ be bipartite with maximum degree at most four, let $\mc L=\emptyset$, $\mc H=\{H_1,H_2\}$, $b_{H_1}\equiv 1$, $b_{H_2}\equiv 1$ and assume that each edge weight is either $1$ or $2$, where $H_i=\{S_i,T_i;E_i\}$ for $i\in\{1,2\}$.
  We show an instance of the weighted $d$-distance matching problem such that there is a one-to-one, weight-preserving mapping between distance matchings and simultaneous assignments.
  Without loss of generality, one can assume that $H_1\not\subseteq H_2$ and $H_2\not\subseteq H_1$.
  To construct an instance $G'=(S',T';E'), d\in\N$ of the weighted $d$-distance matching problem, let $G'=G$ along with the weight of the edges and modify $G'$ as follows.
  Order the nodes of $S'$ such that the nodes of $S_1\setminus S_2$, $S_1\cap S_2$ and $S_2\setminus S_1$ appear in this order (the order of the nodes inside the three sets is arbitrary).
  Insert $|S_1\setminus S_2|$ and $|S_2\setminus S_1|$ new nodes to $S'$ right after the last node of $S_1$ and right after the last node not covered by $S_2$, respectively.
  Finally, set $d=|S|$.
  Figure~\ref{fig:simultan:dmNPCConstr} illustrates the construction for the instance shown in Figure~\ref{fig:simultan:3dimmatchingconstructionWAPX}.
  The blank nodes are the ones added in the last step.
  \begin{figure}
    \centering
    \begin{tikzpicture}[xscale=.8,yscale=.75]
      \SetVertexMath
      \tikzset{VertexStyle/.append style = {minimum size = 18pt,inner sep=0pt}}
      \grEmptyPath[form=1,x=-2,y=3.5,RA=1.2,rotation=0,prefix=x]{2}

      \Vertex[x=-2+1.2+1.2,y=3.5,L=$ $]{d1}
      \Vertex[x=-2+1.2+1.2+.8,y=3.5,L=$ $]{d2}

      \Vertex[x=7.2,y=3.5,L=$ $]{d1}
      \Vertex[x=7.2+.8,y=3.5,L=$ $]{d2}

      \Vertex[x=2.4,y=3.5,L=e_1^{z_1}]{se1}
      \Vertex[x=3.6,y=3.5,L=e_2^{z_1}]{se2}
      \Vertex[x=4.8,y=3.5,L=e_1^{z_2}]{se3}
      \Vertex[x=6,y=3.5,L=e_2^{z_2}]{se4}
      \grEmptyPath[form=1,x=7.7+1.5,y=3.5,RA=1.2,rotation=0,prefix=y]{2}

      \Vertex[x=1.6,y=0,L=e_1^{z_1}]{e0}
      \Vertex[x=3,y=0,L=e_2^{z_1}]{e1}
      \Vertex[x=5.4,y=0,L=e_1^{z_2}]{e3}
      \Vertex[x=6.6,y=0,L=e_2^{z_2}]{e2}

      \draw[dotted,line width=.9] (e0) -- (x0);
      \draw[dotted,line width=.9] (e0) -- (y0);

      \draw[wavy] (e1) -- (x0);
      \draw[wavy] (e1) -- (y1);

      \draw[dashed,line width=.75] (e2) -- (x1);
      \draw[dashed,line width=.75] (e2) -- (y1);

      \draw[dashdotdotted,line width=.9] (e3) -- (y0);
      \draw[dashdotdotted,line width=.9] (e3) -- (x1);

      \draw[] (e0) -- (se2) node[midway,xshift=2mm,yshift=6mm] {$2$};
      \draw[] (e1) -- (se2) node[midway,xshift=2.3mm,yshift=6mm] {$2$};
      \draw[] (e0) -- (se1);
      \draw[] (e2) -- (se4) node[midway,xshift=0mm,yshift=6mm] {$2$};
      \draw[] (e3) -- (se4) node[midway,xshift=0mm,yshift=6mm] {$2$};
      \draw[] (e3) -- (se3);

    \end{tikzpicture}
    \caption{Illustration of the construction given in the proof of Theorem~\ref{thm:simultan:wdmapxc} for the problem instance presented in Figure~\ref{fig:simultan:3dimmatchingconstructionWAPX}.
      There exists a perfect $9$-distance matching if and only if the problem given in Figure~\ref{fig:simultan:3dimmatchingconstructionWAPX} has a feasible solution of size 9.
      The unlabeled edges are of weight 1.}\label{fig:simultan:dmNPCConstr}
  \end{figure}
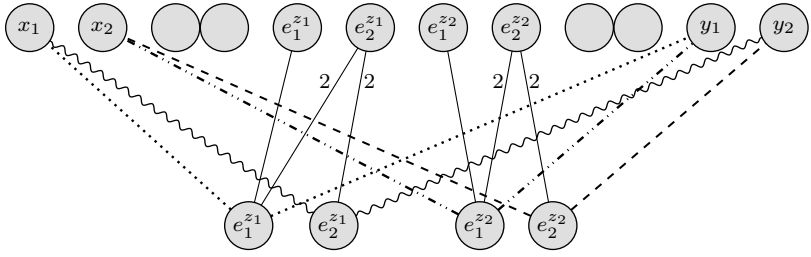
  Observe that there exists an $|S|$-distance matching in $G'$ if and only if there exists a simultaneous assignment $M\subseteq E$ with the same weight, which completes the proof.
  \qed
\end{proof}

As a corollary of Theorem~\ref{thm:simultan:sapAPXhard2}, we prove that the unweighted cyclic $d$-distance matching problem is hard to approximate better than $\frac{760}{759}$.

\begin{theorem}\label{thm:simultan:cdmApxHard}
  The unweighted cyclic distance matching problem is NP-hard to approximate within any factor smaller than $\frac{760}{759}$, even if the maximum degree of the graph is at most $5$.
\end{theorem}
\begin{proof}
  Let us given an instance of the weighted simultaneous assignment problem as per the statement of Theorem~\ref{thm:simultan:sapAPXhard2}, that is, let $G$ be bipartite with maximum degree at most four, $\mc L=\emptyset$, $\mc H=\{H_1,H_2\}$, $b_{H_1}\equiv 1$, $b_{H_2}\equiv 1$ and the weight of each edge is either $1$ or $2$, where $H_i=\{S_i,T_i;E_i\}$ for $i\in\{1,2\}$.

  Let $q$ denote the size of $S_1\setminus S_2$.
  Without loss of generality, assume that $q=|S_1\setminus S_2|=|S_2\setminus S_1|=\frac{|S_1\cap S_2|}{2}$.
  To construct an instance $G'=(S',T';E'),d\in\N$ of the unweighted cyclic $d$-distance matching problem, let $G'=G$ and modify $G'$ as follows.
  Drop the weight of the edges and order the nodes of $S'$ such that the nodes of $S_1\setminus S_2$, $S_1\cap S_2$ and $S_2\setminus S_1$ appear in this order (the order of the nodes inside the three sets is arbitrary).
  Insert a new copy of all nodes of $S_1\cap S_2$ after $S'$.
  These nodes correspond to the hyperedges of the 3-dimensional matching problem, hence let us denote the copies of $e_1^{z}, e_2^{z}$ by $\bar e_1^{z},\bar e_2^{z}$ for $z\in Z$.
  Add $2q$ new dummy nodes to $S'$ 1) after $S'$, 2) after the last node of $S_2$, 3) right after the last node of $S_1$, and 4) right after the last node not contained in $S_2$.
  Also add edges between $\bar e_2^z\in S'$ and $e_i^z\in T$ for all $z\in Z$ and $i\in\{1,2\}$.
  Finally, let $d=5q$.
  Figure~\ref{fig:simultan:cdmApcHard} illustrates the construction.
  The blank nodes are the dummy nodes.
  Note that the maximum degree of $G$ was 4, hence that of $G'$ is at most $5$.

  Let $M$ be a simultaneous assignment in $G$.
  We construct a cyclic $d$-distance matching $M'$ of $G'$ in polynomial time, whose size is the same as the weight of $M$.
  Observe that $M':=M$ is a feasible cyclic $d$-distance matching in $G'$.
  For $z\in Z$, if $M$ contains some edge $e_2^ze_i^z\in S\times T$ with weight $2$, then add edge $\bar e_2^ze_i^z\in S'\times T'$ to $M'$.
  This way, $M'$ remains feasible cyclic $d$-distance matching and the size of $M'$ is the same as the weight of $M$.
  Similarly, given a cyclic distance matching $M'$ in $G'$, one can easily construct a feasible simultaneous assignment $M$ in polynomial time whose weight is (at least) the size of $M'$.

  Hence if there exists an $\alpha$-approximation algorithm for the unweighted cyclic distance matching problem, then there exists an $\alpha$-approximation algorithm for the special case of the simultaneous assignment problem defined in Theorem~\ref{thm:simultan:sapAPXhard2}, which completes the proof.
  \begin{figure}[t]
    \centering

    \begin{tikzpicture}[yscale=.55]
      \SetVertexMath
      \tikzset{VertexStyle/.append style = {minimum size = 18pt,inner sep=0pt}}

      \Vertex[x=-2.8,y=-4,L=e_1^{z_1}]{e0}
      \Vertex[x=-1,y=-4,L=e_2^{z_1}]{e1}
      \Vertex[x=1,y=-4,L=e_1^{z_2}]{e3}
      \Vertex[x=2.8,y=-4,L=e_2^{z_2}]{e2}

      \pgfmathsetmacro\n{28}
      \foreach \i/\k in {
        0/x_1,
        1/x_2,
        2/ ,
        3/ ,
        4/ ,
        5/ ,
        6/e_1^{z_1},
        7/e_2^{z_1},
        8/e_1^{z_2},
        9/e_2^{z_2},
        10/ ,
        11/ ,
        12/ ,
        13/ ,
        14/y_1,
        15/y_2,
        16/ ,
        17/ ,
        18/ ,
        19/ ,
        20/\bar e_2^{z_2},
        21/\bar e_1^{z_2},
        22/\bar e_2^{z_1},
        23/\bar e_1^{z_1},
        24/ ,
        25/ ,
        26/ ,
        27/
      } {
        \pgfmathsetmacro\r{180-\i*(360/\n)+180/\n}
        \pgfmathsetmacro\bx{6*cos(\r}
        \ifthenelse{\i<15\AND \i>0}
        {

          \ifthenelse{\i>3\AND \i<12}
          {
            \pgfmathsetmacro\by{3.1}
          }
          {
            \pgfmathsetmacro\by{5*sin(\r)}
          }
        }
        {
          \ifthenelse{\i>19\AND \i<24}
          {
            \pgfmathsetmacro\by{-6.2}
          }
          {
            \pgfmathsetmacro\by{7*sin(\r)}
          }
        }
        \Vertex[x=\bx,y=\by,L=\k]{v\i};
      }

      \draw[dotted,line width=.9] (e0) -- (v0);
      \draw[dotted,line width=.9] (e0) -- (v14);

      \draw[wavy] (e1) -- (v0);
      \draw[wavy] (e1) -- (v15);

      \draw[dashed,line width=.75] (e2) -- (v1);
      \draw[dashed,line width=.75] (e2) -- (v15);

      \draw[dashdotdotted,line width=.9] (e3) -- (v14);
      \draw[dashdotdotted,line width=.9] (e3) -- (v1);

      \draw[] (e0) -- (v7);
      \draw[] (e1) -- (v7);
      \draw[] (e0) -- (v22);
      \draw[] (e1) -- (v22);
      \draw[] (e0) -- (v6);

      \draw[] (e2) -- (v9);
      \draw[] (e3) -- (v9);
      \draw[] (e2) -- (v20);
      \draw[] (e3) -- (v20);
      \draw[] (e3) -- (v8);

      \begin{pgfonlayer}{background}
        \highlight{2mm}{black!30}{(e0.center) -- (v0.center)}
        \highlight{2mm}{black!30}{(e0.center) -- (v14.center)}
        \highlight{2mm}{black!30}{(e2.center) -- (v1.center)}
        \highlight{2mm}{black!30}{(e2.center) -- (v15.center)}

        \highlight{2mm}{black!30}{(v7.center) -- (e1.center)}
        \highlight{2mm}{black!30}{(v9.center) -- (e3.center)}
      \end{pgfonlayer}
    \end{tikzpicture}
    \caption{Illustration of the construction in the proof of Theorem~\ref{thm:simultan:cdmApxHard} for the problem instance presented in Figure~\ref{fig:simultan:3dimMatchingConstruction}.
      The nodes of $S$ are the outer nodes.
      They are considered in cyclic order with $d=5|S_1\setminus S_2|=10$.
      The inner nodes are the nodes of $T$.}\label{fig:simultan:cdmApcHard}
  \end{figure}
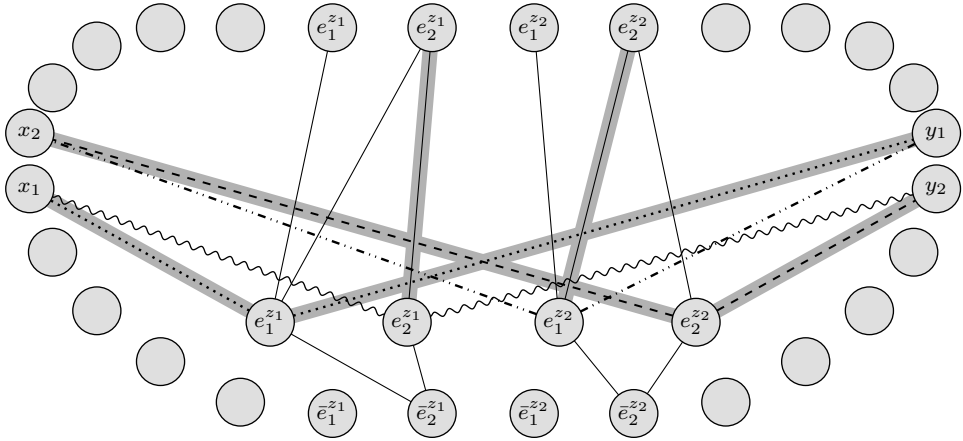
  \qed
\end{proof}

\section{Approximation Algorithms}\label{sec:simultan:apx}
Throughout this section, let $k=|\mc H|$, $\mc H=\{H_1,\dots,H_k\}$ and $H_i=(V_{i},E_{i})$ for $i\in\{1,\dots,k\}$.
Without loss of generality, assume that $\mc H\neq\emptyset$.
First, a general framework is given for deriving approximation algorithms for the simultaneous assignment problem, which will be utilized in the rest of the section in multiple settings.
The following definition plays a central role:

\begin{definition}\label{def:simultan:mlCover}
  Given an instance of the simultaneous assignment problem, we call $m$ not necessarily distinct subsets $F_1,\dots,F_m$ of the edges an \emph{$(m,\ell)$-cover} if every edge of $G$ is contained in at least $\ell$ of $F_1,\dots,F_m$.
\end{definition}

\begin{theorem}\label{thm:simultan:apxGapBound}
  Given a linear program whose integer solutions are exactly the feasible simultaneous assignments, let $F_1,\dots,F_m$ be an $(m,\ell)$-cover of $G=(V,E)$ such that the linear program becomes integer when restricted to $F_i$ for all $i\in\{1,\dots,m\}$.
  Then, the integrality gap of the linear program is at most $\frac{m}{\ell}$.
\end{theorem}
\begin{proof}
  Let $x$ be an optimal solution to the linear program and let $z$ be an optimal integer solution.
  Furthermore, let $x_i$ and $z_i$ denote an optimal fractional and integer solution to the problem restricted to $F_i$ for $i\in\{1,\dots,m\}$.
  Note that these solutions are also feasible solutions to the original problem.
  The following computation shows that the integrality gap is at most $\frac{m}{\ell}$.
  \begin{multline}\label{eq:simultan:apxMain:comp}
    \ell wx \leq  \sum_{i=1}^m\sum_{e\in F_i} w(e)x(e)\leq\sum_{i=1}^m\sum_{e\in F_i} w(e)x_i(e)\\
    \leq m\sum_{e\in F_{i^*}} w(e)x_{i^*}(e)=m\sum_{e\in F_{i^*}} w(e)z_{i^*}(e)\leq mwz,
  \end{multline}
  where $i^*=\argmax_{i\in\{1,\dots,m\}}\{\sum_{e\in F_i} w(e)x_i(e)\}$.
  The first inequality holds because every edge of $G$ is contained in at least $\ell$ of $F_1,\dots,F_m$, the second one follows by the optimality of $x_i$ for $F_i$, the third one by the selection of $i^*$, whereas the equation holds because the linear program is integer when the problem is restricted to $F_{i^*}$.
  By~(\ref{eq:simultan:apxMain:comp}), one gets that $\frac{wx}{wz}\leq\frac{m}{\ell }$, which completes the proof.
  \qed
\end{proof}

Theorem~\ref{thm:simultan:apxGapBound} gives a general framework for deriving bounds on the integrality gap of integer programs.
Note that the proof of Theorem~\ref{thm:simultan:apxGapBound}
can be turned into an approximation algorithm if
an $(m,\ell)$-cover $F_1,\dots,F_m$ is given and the linear program can be solved efficiently when the problem is restricted to $F_i$ for all $i\in\{1,\dots,m\}$ and $m$ is polynomial in the size of the problem.
In fact, one can avoid linear programming altogether and obtain an efficient $\frac{m}{\ell}$-approximation algorithm provided that the problems restricted to $F_i$ are tractable and the heaviest among them can be found in polynomial time.

\begin{theorem}\label{thm:simultan:apxComb}
  Let $F_1,\dots,F_m$ be an $(m,\ell)$-cover.
  Then, the heaviest among the optimal simultaneous assignments in the problems restricted to the edge set $F_i$ is an $\frac{m}{\ell}$-approximate solution for the original problem.
\end{theorem}
\begin{proof}
  Let $M_i$ denote an optimal solution to the problem restricted to $F_i$ and let $M^*$ be an optimal simultaneous assignment in $G$.
  Then,
  \begin{equation}
    \ell w(M^*)\leq\sum_{i=1}^m\sum_{e\in F_i\cap M^*}w(e)\leq\sum_{i=1}^m\sum_{e\in M_i}w(e)\leq m\sum_{e\in M_{i^*}}w(e)=mw(M_{i^*})
  \end{equation}
  holds, where $i^*=\argmax\{w(M_i) : i\in\{1,\dots,m\}\}$.
  This means that $\frac{w(M^*)}{w(M_{i^*})}\leq\frac{m}{\ell}$, that is,  $M_{i^*}$ is indeed $\frac{m}{\ell}$-approximate.
  Finally, observe that $M_{i^*}$ is a feasible solution to the original problem.
  \qed
\end{proof}

Theorem~\ref{thm:simultan:apxComb} gives a framework for deriving approximation algorithms for the simultaneous assignment problem.
Namely, we need to find an $(m,\ell)$-cover $F_1,\dots,F_m$ --- trying to minimize the ratio $\frac{m}{\ell}$ --- such that one can find the best among the optimal solutions to the problems restricted to $F_i$ for $i\in\{1,\dots,m\}$, which is an $\frac{m}{\ell}$-approximate solution by Theorem~\ref{thm:simultan:apxComb}.

In fact, this framework easily extends to problems other than the simultaneous assignment problem --- the main requirement is that any subset of a feasible solution should be feasible.

\subsection{Approximation Algorithm for Small $k$}\label{sec:simultan:apx:constAPX}

This section applies the approximation framework above to the simultaneous assignment problem in the case when the size of $\mc H$ is small.
Throughout this section, let $k=|\mc H|$ and let $k'\in\{1,\dots,k\}$ be the smallest integer for which every edge appears in at most $k'$ of the subgraphs in $\mc H$.
Consider the following type of $(m,\ell)$-covers.
\begin{definition}
  An $(m,\ell)$-cover $F_1,\dots,F_m$ is \emph{laminar} if the problem restricted to $F_i$ is laminar for each $i\in\{1,\dots,m\}$.
\end{definition}

In the light of Theorems~\ref{thm:simultan:apxGapBound}~and~\ref{thm:simultan:apxComb}, we want to construct a laminar $(m,\ell)$-cover minimizing the value $\frac{m}{\ell}$.
Let $\alpha(k,k')$ denote the minimum value of $\frac{m}{\ell}$ for which a laminar $(m',\ell')$-cover always exists such that $\frac{m'}{\ell'}\leq\frac{m}{\ell}$ whenever $k=|\mc H|$ and every edge appears in at most $k'$ subgraphs in $\mc H$.
In other words, $\alpha(k,k')$ is the best approximation ratio one can hope for by applying Theorem~\ref{thm:simultan:apxComb} to a laminar cover.
The following min-max theorem gives an easy-to-compute formula for $\alpha(k,k')$.

\begin{theorem}\label{thm:simultan:laminarCoverMain}
  Let $k$ and $k'$ be as above.
  The minimum value of $\frac{m}{\ell}$ for which there always exists a laminar $(m',\ell')$-cover such that $\frac{m'}{\ell'}\leq\frac{m}{\ell}$, that is, $\alpha(k,k')$, equals
  \begin{equation}\label{eq:simultan:alphaVal}
    \max_{j\in\{0,\dots,k'-1\}}\frac{1}{k-j}\sum_{i=j+1}^{k'}\binom{k}{i}.
  \end{equation}
  Furthermore, an $\alpha(k,k')$-approximate solution can be found in $\mathcal{O}(f(k)\poly(|V|,|E|))$ steps.
\end{theorem}
\begin{proof}
  If either $k=1$ or $k'=1$, then the problem is laminar, hence $E$ itself is a laminar $(1,1)$-cover.
  This means that $\alpha(k,1)=\alpha(1,k')=1$, which matches the value given by~(\ref{eq:simultan:alphaVal}).
  Therefore, one can assume that $k,k'\geq 2$.

  Let us partition the edge set based on which of the subgraphs in $\mc H$ they are included.
  Formally, let
  \begin{equation}\label{eq:simultan:catDef}
    C_I=\bigcap_{i\in I}E_i\cap\bigcap_{j\in\{1,\dots,k\}\setminus I}E\setminus E_j
  \end{equation}
  for all $I\subseteq\{1,\dots,k\}$.
  We will refer to $\mc C_I$ as an \emph{edge category} for each $I\subseteq\{1,\dots,k\}$.
  By definition, $C_I=\emptyset$ if $|I|>k'$.
  On the other hand, one can assume without loss of generality that $C_I\neq\emptyset$ when $|I|\leq k'$.
  Indeed, if there were such an empty category $C_I$, one could add two new nodes to $G$ and an edge $e$ of weight zero between them which falls into category $C_I$ --- clearly, a laminar $(m,\ell)$-cover of the resulting problem is also a laminar $(m,\ell)$-cover of the original problem if $e$ is left out and vice versa.
  The set of all non-empty edge categories will be denoted by
  \begin{equation}
    \mc C =\{C_I : I\subseteq\{1,\dots,k\}, C_I\neq\emptyset\}.
  \end{equation}
  We say that a subset $\mc F$ of $\mc C$ is a \emph{laminar category system} if the problem becomes laminar when restricted to the edges of the union of the categories in $\mc F$.
  Let $\mc L_{\mc C}$ denote the set of all \emph{maximal} laminar category systems.
  Let $\alpha'(k,k')$ denote the minimum value of $\frac{m}{\ell}$ for which there always exists a laminar $(m,\ell)$-cover $F'_1,\dots,F'_k$ such that \emph{each $F_i'$ is the union of the categories in a laminar category system}.
  \begin{lemma}\label{lem:simultan:categorysetsonly}
    $\alpha'(k,k')=\alpha(k,k')$.
  \end{lemma}
  \begin{proof}
    By definition, $\alpha'(k,k')\geq\alpha(k,k')$, since the set of laminar $(m,\ell)$-covers formed from laminar category systems is a subset of all possible laminar $(m,\ell)$-covers.
    To see the equation, let $\tilde G=(\tilde V,\tilde E)$ be a star graph with $2^k$ edges, where $\tilde E=\{e_0,\dots,e_{2^k-1}\}$.
    Let $\mc H=\{H_1,\dots,H_k\}$ be such that $\mc C$ consists of singletons, that is, each edge category consists of a single edge.
    One can easily construct such $\mc H$ as follows.
    Let $e_i\in H_j$ if and only if the $j^{\text{th}}$ bit of $i$ in base-2 is one.
    For example, $e_0$ is in no subgraphs, $e_1$ is in $H_1$, $e_2$ is in $H_2$, $e_3$ is both in $H_1$ and $H_2$, etc.
    Now, remove every edge which appears in more than $k'$ categories.
    For $\tilde G$, the minimum value $\frac{\tilde m}{\tilde \ell}$ for which a laminar $(\tilde m,\tilde \ell)$-cover $F_1,\dots,F_{\tilde m}$ exists does not change if we require that each of $F_1,\dots,F_{\tilde m}$ is the union of the categories in a laminar category system, as laminarity and locally laminarity mean the same for $\tilde G$.
    Notice that if a subset of categories is locally laminar for this instance, then the same selection of categories is also locally laminar for any other instance of the simultaneous assignment problem when no edge appears in more than $k'$ subgraphs in $\mc H$.
    This implies that $\alpha'(k,k')=\frac{\tilde m}{\tilde \ell}=\alpha(k,k')$, which was to be shown.
    \qed
  \end{proof}

  By Lemma~\ref{lem:simultan:categorysetsonly}, we can assume that our laminar $(m,\ell)$-cover consists of edge sets which are the unions of the categories in a laminar category system.
  In fact, we can restrict ourselves to \emph{maximal} laminar category systems.
  An optimal laminar $(m,\ell)$-cover using maximal laminar category systems can be represented as a vector $z\in\Z_+^{\mc L_{\mc C}}$ and an integer $\ell\in\N$ such that
  \begin{equation}\label{eq:simultan:laminarCoverInteger}
    \sum_{\mc F : C\in\mc F\in\mc L_{\mc C}}z_{\mtiny{\mc F}}\geq \ell
  \end{equation}
  holds for each category $C\in\mc C$ and
  \begin{equation}\label{eq:simultan:laminarCoverIntegerObjective}
    \frac{1}{\ell}\sum_{\mc F\in\mc L_{\mc C}}z_{\mtiny{\mc F}}
  \end{equation}
  is minimized, where $\mc L_{\mc C}$ denotes the set of all maximal laminar category systems, $z_{{\mtiny{\mc F}}}$ is the multiplicity of category system $\mc F\in\mc L_{\mc C}$ in the cover and $\ell$ is such that the cover defined by $z$ covers all categories (and hence the edges in them) at least $\ell$ times.
  To find such $z$ and $\ell$, it suffices to solve the following
  linear program.
  \begin{subequations}
    \begin{align}\label{lp:simultan:laminarCoverLp}
      \tag{LP7}
      \min&\sum_{\mc F\in\mc L_{\mc C}}y_{\mtiny{\mc F}}\\
      \mbox{s.t.}\quad\quad\quad\quad\quad&&\nonumber\\
      y&\in\Q_+^{\mc L_{\mc C}}&\\
      \sum_{\mc F : C\in \mc F\in\mc L_{\mc C}}y_{\mtiny{\mc F}} &\geq 1 &\forall C\in\mc C\label{eq:simultan:laminarCoverLP:ineq}
    \end{align}
  \end{subequations}
  Given an optimal solution $y\in\Q_+^{\mc L_{\mc C}}$ to (\ref{lp:simultan:laminarCoverLp}), one can construct a solution $z\in\Z_+^{\mc L_{\mc C}}$ to~(\ref{eq:simultan:laminarCoverInteger}) which minimizes~(\ref{eq:simultan:laminarCoverIntegerObjective}) as follows.
  Let $z:=\ell y$, where $\ell\in\N$ is the lowest common denominator of the coordinates of $y$.
  For this $z\in\Z_+^{\mc L_{\mc C}}$ and $\ell\in\N$, (\ref{eq:simultan:laminarCoverInteger}) holds, which means that $z$ corresponds to a laminar $(m,\ell)$-cover, where $m=\ell\sum_{\mc F\in\mc L_{\mc C}}y_{\mtiny{\mc F}}$, and hence $\alpha(k,k')=\sum_{\mc F\in\mc L_{\mc C}}y_{\mtiny{\mc F}}$ --- which was minimized in~(\ref{lp:simultan:laminarCoverLp}).

  Our aim is to construct an optimal solution to (\ref{lp:simultan:laminarCoverLp}).
  First, consider the following representation of maximal laminar category systems, which resembles the usual arborescence-representation of laminar families~\cite[Page 36]{AF11}.
  \begin{lemma}\label{lem:simultan:maxLaminCatSysRepr}
    There is a one-to-one correspondence between maximal laminar category systems of $\mc C$ and labeled trees on nodes $v_0,\dots,v_k$ in which the maximum distance from $v_0$ is at most $k'$.
  \end{lemma}
  \begin{proof}
    Let $\mc F\in\mc L_{\mc C}$ be a maximal laminar category system, and let $\mc I = \{I : C_I\in\mc F\}$.
    Let $E'_i$ denote those edges in $E_i$ which are contained in some of the categories in $\mc F$, where $i\in\{1,\dots,k\}$.
    We also define $E_0=C_\emptyset=E\setminus\bigcup E'_i$.
    By definition, the system formed by $E'_0,\dots,E'_k$ is laminar, meaning that $E'_i$ and $E'_j$ are either disjoint or one of them contains the other.
    The former case holds if and only if there exists no $I\in\mc I$ such that $\{i,j\}\subseteq I$.
    In the latter case, when one of  $E'_i$ and $E'_j$ contains the other, one can assume without loss of generality that $E'_i\subseteq E'_j$.
    By~(\ref{eq:simultan:catDef}), this holds if and only if there exists no $I\in\mc I$ such that $i\in I$ and $j\notin I$.

    Let $\mc I_i=\{I\in\mc I : i\in I\}$ for $i\in\{1,\dots,k\}$, and let $\mc I_0=\mc I$.
    By the properties above, $\{ \mc I_i : i\in\{0,\dots,k\}\}$ is a laminar system on ground set $\mc I$.
    We define a tree $T$ representing this system (and hence $\mc F$) on nodes $\{v_0,\dots,v_k\}$ as follows.
    Let $v_i$ be associated with $\mc I_i$ for $i\in\{0,\dots,k\}$.
    Connect $v_i$ and $v_j$ if $\mc I_j$ is maximal inside $\mc I_i$.
    Figure~\ref{fig:simultan:treeReprOfMaxLamSystems} illustrates the construction for $k=k'=3$.
    It represents all maximal laminar category systems as Venn-diagrams and also shows their tree representations.
    Note that the maximum distance from $v_0$ is at most $k'$ since every $I\in\mc I$ appears in at most $k'$ of $\mc I_1,\dots,\mc I_k$.
    \qed
  \end{proof}
  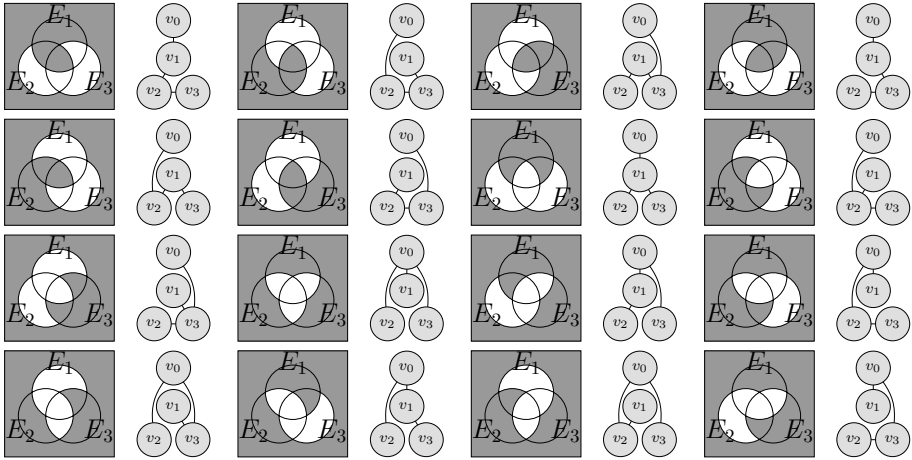
\begin{figure}[!h]              
    \centering
    \scalebox{.71}{
      \begin{tabular}{c c c c c c c c c c c}
        \begin{tikzpicture}[x=1.7em,y=1.7em]
          \pic at (0,0) {venn 3/.expanded=10010011};
        \end{tikzpicture}
        \hspace{-2mm}
        & \begin{tikzpicture}[x=1.2em,y=1.2em]
          \SetVertexMath
          dot/.style = {circle, draw, semithick,
            inner sep=0pt, minimum size=1pt,
            node contents={}},
          \tikzset{edge/.style = {->,> = latex'}}
          \Vertex[x=0, y=2 + sqrt 3,L={v_0}]{v_0}
          \Vertex[x=0, y=sqrt 3,    L={v_1}]{v_1}
          \Vertex[x=-1,y=0,         L={v_2}]{v_2}
          \Vertex[x=1, y=0,         L={v_3}]{v_3}
          \draw[] (v_0) -- (v_1);
          \draw[] (v_1) -- (v_2);
          \draw[] (v_2) -- (v_3);
        \end{tikzpicture}
        & \begin{tikzpicture}[x=1.7em,y=1.7em]
          \pic at (0,0) {venn 3/.expanded=10100101};
        \end{tikzpicture}
          \hspace{-2mm}
        & \begin{tikzpicture}[x=1.2em,y=1.2em]
          \SetVertexMath
          dot/.style = {circle, draw, semithick,
            inner sep=0pt, minimum size=1pt,
            node contents={}},
          \tikzset{edge/.style = {->,> = latex'}}
          \Vertex[x=0, y=2 + sqrt 3,L={v_0}]{v_0}
          \Vertex[x=0, y=sqrt 3,    L={v_1}]{v_1}
          \Vertex[x=-1,y=0,         L={v_2}]{v_2}
          \Vertex[x=1, y=0,         L={v_3}]{v_3}
          \draw[] (v_0) to[bend right = 20pt] (v_2);
          \draw[] (v_2) -- (v_3);
          \draw[] (v_3) -- (v_1);
        \end{tikzpicture}
        & \begin{tikzpicture}[x=1.7em,y=1.7em]
          \pic at (0,0) {venn 3/.expanded=11001001};
        \end{tikzpicture}
          \hspace{-2mm}
        & \begin{tikzpicture}[x=1.2em,y=1.2em]
          \SetVertexMath
          dot/.style = {circle, draw, semithick,
            inner sep=0pt, minimum size=1pt,
            node contents={}},
          \tikzset{edge/.style = {->,> = latex'}}
          \Vertex[x=0, y=2 + sqrt 3,L={v_0}]{v_0}
          \Vertex[x=0, y=sqrt 3,    L={v_1}]{v_1}
          \Vertex[x=-1,y=0,         L={v_2}]{v_2}
          \Vertex[x=1, y=0,         L={v_3}]{v_3}
          \draw[] (v_0) to[bend left = 20pt] (v_3);
          \draw[] (v_3) -- (v_1);
          \draw[] (v_1) -- (v_2);
        \end{tikzpicture}
        & \begin{tikzpicture}[x=1.7em,y=1.7em]
          \pic at (0,0) {venn 3/.expanded=10011001};
        \end{tikzpicture}
          \hspace{-2mm}
        & \begin{tikzpicture}[x=1.2em,y=1.2em]
          \SetVertexMath
          dot/.style = {circle, draw, semithick,
            inner sep=0pt, minimum size=1pt,
            node contents={}},
          \tikzset{edge/.style = {->,> = latex'}}
          \Vertex[x=0, y=2 + sqrt 3,L={v_0}]{v_0}
          \Vertex[x=0, y=sqrt 3,    L={v_1}]{v_1}
          \Vertex[x=-1,y=0,         L={v_2}]{v_2}
          \Vertex[x=1, y=0,         L={v_3}]{v_3}
          \draw[] (v_0) -- (v_1);
          \draw[] (v_1) -- (v_3);
          \draw[] (v_3) -- (v_2);
        \end{tikzpicture}
        \\
        \begin{tikzpicture}[x=1.7em,y=1.7em]
          \pic at (0,0) {venn 3/.expanded=10100011};
        \end{tikzpicture}
        \hspace{-2mm}
        & \begin{tikzpicture}[x=1.2em,y=1.2em]
          \SetVertexMath
          dot/.style = {circle, draw, semithick,
            inner sep=0pt, minimum size=1pt,
            node contents={}},
          \tikzset{edge/.style = {->,> = latex'}}
          \Vertex[x=0, y=2 + sqrt 3,L={v_0}]{v_0}
          \Vertex[x=0, y=sqrt 3,    L={v_1}]{v_1}
          \Vertex[x=-1,y=0,         L={v_2}]{v_2}
          \Vertex[x=1, y=0,         L={v_3}]{v_3}
          \draw[] (v_0) to[bend right = 20pt] (v_2);
          \draw[] (v_2) -- (v_1);
          \draw[] (v_1) -- (v_3);
        \end{tikzpicture}
        & \begin{tikzpicture}[x=1.7em,y=1.7em]
          \pic at (0,0) {venn 3/.expanded=11000101};
        \end{tikzpicture}
          \hspace{-2mm}
        & \begin{tikzpicture}[x=1.2em,y=1.2em]
          \SetVertexMath
          dot/.style = {circle, draw, semithick,
            inner sep=0pt, minimum size=1pt,
            node contents={}},
          \tikzset{edge/.style = {->,> = latex'}}
          \Vertex[x=0, y=2 + sqrt 3,L={v_0}]{v_0}
          \Vertex[x=0, y=sqrt 3,    L={v_1}]{v_1}
          \Vertex[x=-1,y=0,         L={v_2}]{v_2}
          \Vertex[x=1, y=0,         L={v_3}]{v_3}
          \draw[] (v_0) to[bend left = 20pt] (v_3);
          \draw[] (v_3) -- (v_2);
          \draw[] (v_2) -- (v_1);
        \end{tikzpicture}
        & \begin{tikzpicture}[x=1.7em,y=1.7em]
          \pic at (0,0) {venn 3/.expanded=10011010};
        \end{tikzpicture}
          \hspace{-2mm}
        & \begin{tikzpicture}[x=1.2em,y=1.2em]
          \SetVertexMath
          dot/.style = {circle, draw, semithick,
            inner sep=0pt, minimum size=1pt,
            node contents={}},
          \tikzset{edge/.style = {->,> = latex'}}
          \Vertex[x=0, y=2 + sqrt 3,L={v_0}]{v_0}
          \Vertex[x=0, y=sqrt 3,    L={v_1}]{v_1}
          \Vertex[x=-1,y=0,         L={v_2}]{v_2}
          \Vertex[x=1, y=0,         L={v_3}]{v_3}
          \draw[] (v_0) -- (v_1);
          \draw[] (v_1) -- (v_2);
          \draw[] (v_1) -- (v_3);
        \end{tikzpicture}
        & \begin{tikzpicture}[x=1.7em,y=1.7em]
          \pic at (0,0) {venn 3/.expanded=10100110};
        \end{tikzpicture}
          \hspace{-2mm}
        & \begin{tikzpicture}[x=1.2em,y=1.2em]
          \SetVertexMath
          dot/.style = {circle, draw, semithick,
            inner sep=0pt, minimum size=1pt,
            node contents={}},
          \tikzset{edge/.style = {->,> = latex'}}
          \Vertex[x=0, y=2 + sqrt 3,L={v_0}]{v_0}
          \Vertex[x=0, y=sqrt 3,    L={v_1}]{v_1}
          \Vertex[x=-1,y=0,         L={v_2}]{v_2}
          \Vertex[x=1, y=0,         L={v_3}]{v_3}
          \draw[] (v_0) to[bend right = 20pt] (v_2);
          \draw[] (v_2) -- (v_1);
          \draw[] (v_2) -- (v_3);
        \end{tikzpicture}
        \\
        \begin{tikzpicture}[x=1.7em,y=1.7em]
          \pic at (0,0) {venn 3/.expanded=11001100};
        \end{tikzpicture}
        \hspace{-2mm}
        & \begin{tikzpicture}[x=1.2em,y=1.2em]
          \SetVertexMath
          dot/.style = {circle, draw, semithick,
            inner sep=0pt, minimum size=1pt,
            node contents={}},
          \tikzset{edge/.style = {->,> = latex'}}
          \Vertex[x=0, y=2 + sqrt 3,L={v_0}]{v_0}
          \Vertex[x=0, y=sqrt 3,    L={v_1}]{v_1}
          \Vertex[x=-1,y=0,         L={v_2}]{v_2}
          \Vertex[x=1, y=0,         L={v_3}]{v_3}
          \draw[] (v_0) to[bend left = 20pt] (v_3);
          \draw[] (v_3) -- (v_1);
          \draw[] (v_3) -- (v_2);
        \end{tikzpicture}
        & \begin{tikzpicture}[x=1.7em,y=1.7em]
          \pic at (0,0) {venn 3/.expanded=11110000};
        \end{tikzpicture}
          \hspace{-2mm}
        & \begin{tikzpicture}[x=1.2em,y=1.2em]
          \SetVertexMath
          dot/.style = {circle, draw, semithick,
            inner sep=0pt, minimum size=1pt,
            node contents={}},
          \tikzset{edge/.style = {->,> = latex'}}
          \Vertex[x=0, y=2 + sqrt 3,L={v_0}]{v_0}
          \Vertex[x=0, y=sqrt 3,    L={v_1}]{v_1}
          \Vertex[x=-1,y=0,         L={v_2}]{v_2}
          \Vertex[x=1, y=0,         L={v_3}]{v_3}
          \draw[] (v_0) -- (v_1);
          \draw[] (v_0) to[bend left = 20pt] (v_3);
          \draw[] (v_0) to[bend right = 20pt] (v_2);
        \end{tikzpicture}
        & \begin{tikzpicture}[x=1.7em,y=1.7em]
          \pic at (0,0) {venn 3/.expanded=11010010};
        \end{tikzpicture}
          \hspace{-2mm}
        & \begin{tikzpicture}[x=1.2em,y=1.2em]
          \SetVertexMath
          dot/.style = {circle, draw, semithick,
            inner sep=0pt, minimum size=1pt,
            node contents={}},
          \tikzset{edge/.style = {->,> = latex'}}
          \Vertex[x=0, y=2 + sqrt 3,L={v_0}]{v_0}
          \Vertex[x=0, y=sqrt 3,    L={v_1}]{v_1}
          \Vertex[x=-1,y=0,         L={v_2}]{v_2}
          \Vertex[x=1, y=0,         L={v_3}]{v_3}
          \draw[] (v_0) -- (v_1);
          \draw[] (v_0) to[bend left = 20pt] (v_3);
          \draw[] (v_1) -- (v_2);
        \end{tikzpicture}
        & \begin{tikzpicture}[x=1.7em,y=1.7em]
          \pic at (0,0) {venn 3/.expanded=10110100};
        \end{tikzpicture}
          \hspace{-2mm}
        & \begin{tikzpicture}[x=1.2em,y=1.2em]
          \SetVertexMath
          dot/.style = {circle, draw, semithick,
            inner sep=0pt, minimum size=1pt,
            node contents={}},
          \tikzset{edge/.style = {->,> = latex'}}
          \Vertex[x=0, y=2 + sqrt 3,L={v_0}]{v_0}
          \Vertex[x=0, y=sqrt 3,    L={v_1}]{v_1}
          \Vertex[x=-1,y=0,         L={v_2}]{v_2}
          \Vertex[x=1, y=0,         L={v_3}]{v_3}
          \draw[] (v_0) -- (v_1);
          \draw[] (v_0) to[bend right = 20pt] (v_2);
          \draw[] (v_2) -- (v_3);
        \end{tikzpicture}
        \\
        \begin{tikzpicture}[x=1.7em,y=1.7em]
          \pic at (0,0) {venn 3/.expanded=11101000};
        \end{tikzpicture}
        \hspace{-2mm}
        & \begin{tikzpicture}[x=1.2em,y=1.2em]
          \SetVertexMath
          dot/.style = {circle, draw, semithick,
            inner sep=0pt, minimum size=1pt,
            node contents={}},
          \tikzset{edge/.style = {->,> = latex'}}
          \Vertex[x=0, y=2 + sqrt 3,L={v_0}]{v_0}
          \Vertex[x=0, y=sqrt 3,    L={v_1}]{v_1}
          \Vertex[x=-1,y=0,         L={v_2}]{v_2}
          \Vertex[x=1, y=0,         L={v_3}]{v_3}
          \draw[] (v_0) to[bend right = 20pt] (v_2);
          \draw[] (v_0) to[bend left = 20pt] (v_3);
          \draw[] (v_3) -- (v_1);
        \end{tikzpicture}
        &\begin{tikzpicture}[x=1.7em,y=1.7em]
          \pic at (0,0) {venn 3/.expanded=10111000};
        \end{tikzpicture}
          \hspace{-2mm}
        & \begin{tikzpicture}[x=1.2em,y=1.2em]
          \SetVertexMath
          dot/.style = {circle, draw, semithick,
            inner sep=0pt, minimum size=1pt,
            node contents={}},
          \tikzset{edge/.style = {->,> = latex'}}
          \Vertex[x=0, y=2 + sqrt 3,L={v_0}]{v_0}
          \Vertex[x=0, y=sqrt 3,    L={v_1}]{v_1}
          \Vertex[x=-1,y=0,         L={v_2}]{v_2}
          \Vertex[x=1, y=0,         L={v_3}]{v_3}
          \draw[] (v_0) -- (v_1);
          \draw[] (v_0) to[bend right = 20pt] (v_2);
          \draw[] (v_1) -- (v_3);
        \end{tikzpicture}
        & \begin{tikzpicture}[x=1.7em,y=1.7em]
          \pic at (0,0) {venn 3/.expanded=11100010};
        \end{tikzpicture}
          \hspace{-2mm}
        & \begin{tikzpicture}[x=1.2em,y=1.2em]
          \SetVertexMath
          dot/.style = {circle, draw, semithick,
            inner sep=0pt, minimum size=1pt,
            node contents={}},
          \tikzset{edge/.style = {->,> = latex'}}
          \Vertex[x=0, y=2 + sqrt 3,L={v_0}]{v_0}
          \Vertex[x=0, y=sqrt 3,    L={v_1}]{v_1}
          \Vertex[x=-1,y=0,         L={v_2}]{v_2}
          \Vertex[x=1, y=0,         L={v_3}]{v_3}
          \draw[] (v_0) to[bend right = 20pt] (v_2);
          \draw[] (v_0) to[bend left = 20pt] (v_3);
          \draw[] (v_2) -- (v_1);
        \end{tikzpicture}
        & \begin{tikzpicture}[x=1.7em,y=1.7em]
          \pic at (0,0) {venn 3/.expanded=11010100};
        \end{tikzpicture}
          \hspace{-2mm}
        & \begin{tikzpicture}[x=1.2em,y=1.2em]
          \SetVertexMath
          dot/.style = {circle, draw, semithick,
            inner sep=0pt, minimum size=1pt,
            node contents={}},
          \tikzset{edge/.style = {->,> = latex'}}
          \Vertex[x=0, y=2 + sqrt 3,L={v_0}]{v_0}
          \Vertex[x=0, y=sqrt 3,    L={v_1}]{v_1}
          \Vertex[x=-1,y=0,         L={v_2}]{v_2}
          \Vertex[x=1, y=0,         L={v_3}]{v_3}
          \draw[] (v_0) -- (v_1);
          \draw[] (v_0) to[bend left = 20pt] (v_3);
          \draw[] (v_3) -- (v_2);
        \end{tikzpicture}

      \end{tabular}
    }
    \caption{Tree representation of all maximal laminar category systems for $k=k'=3$.}\label{fig:simultan:treeReprOfMaxLamSystems}
  \end{figure}
  \begin{remark}\label{rem:simultan:repr}
    In the correspondence between maximal laminar category systems and trees on nodes $v_0,\dots,v_k$ described in the proof of Lemma~\ref{lem:simultan:maxLaminCatSysRepr}, a node $v$ of a tree $T$ represents a category $I\in\mc C$ whose size is the distance between $v_0$ and $v$, furthermore, the represented category only depends on the path between $v_0$ and $v$.
  \end{remark}

  In the light of Lemma~\ref{lem:simultan:maxLaminCatSysRepr}, one may think of maximal laminar category systems as labeled trees on $(k+1)$ nodes with (rooted) depth of at most $k'$.
  For $j\in\{1,\dots,k'\}$, let $\mc T_j$ be the set of labeled trees on node set $v_0,\dots,v_k$ each of which is the union of a path $v_0,\dots,u$ of length $(j-1)$ and every other node is connected to $u$.
  Note that by Lemma~\ref{lem:simultan:maxLaminCatSysRepr}, these trees correspond to maximal laminar category systems --- for $k=k'=3$, the first six diagrams in Figure~\ref{fig:simultan:treeReprOfMaxLamSystems} (read row-by-row) correspond to $\mc T_3$, the next three to $\mc T_2$ and the next one to $\mc T_1$.
  In what follows, we construct a laminar category cover that consists of maximal laminar category systems represented by the trees in $\mc T_1,\dots,\mc T_{k'}$, and we show that it is an optimal solution to~(\ref{lp:simultan:laminarCoverLp}).
  To this end, the next lemma counts the vertices of these trees which represent a given category $C$.
  \begin{lemma}\label{lem:simultan:laminCoveringNumber}
    The laminar category systems represented by the trees in $\mc T_j$ contain every category in $\mc C_i$ exactly $a_{ij}$ times, where
    \begin{equation}\label{eq:simultan:aijDef}
      a_{ij}=\begin{cases}
        \frac{(k-i)!i!}{(k-j+1)!} & \text{if } i<j,\\
        j! & \text{if } i=j,\\
        0& \text{otherwise}\\
      \end{cases}
    \end{equation}
    for each $i,j\in\{1,\dots,k'\}$.
  \end{lemma}
  \begin{proof}
    By Remark~\ref{rem:simultan:repr} and by symmetry, every category in $\mc C_i$ is represented exactly the same number of times in $\mc T_j$.
    Therefore, to derive the value of $a_{ij}$, it suffices to count the nodes of the trees in $\mc T_j$ which represent a category in $\mc C_i$ and divide this number by the cardinality of $\mc C_i$.

    Case 1.
    Suppose that $i<j$.
    Then, every tree in $\mc T_j$ contains exactly one node at distance $i$, and hence each tree represents a laminar system containing exactly one category from $\mc C_i$.
    For every path of length $(j-1)$ on $(k+1)$ nodes with one of its endpoints fixed, exactly one tree has been added to $\mc T_j$, therefore $|\mc T_j|=\frac{k!}{(k-j+1)!}$.
    From this, one gets that
    \begin{equation*}
      a_{ij}=\frac{|\mc T_j|}{|\mc C_i|}=\frac{\frac{k!}{(k-j+1)!}}{\binom{k}{i}}=\frac{(k-i)!i!}{(k-j+1)!},
    \end{equation*}
    which is identical to the first case of~(\ref{eq:simultan:aijDef}).

    Case 2.
    Suppose that $i=j$.
    Then, every tree of $\mc T_j$ contains exactly $(k-j+1)$ nodes at distance $i$ from $v_0$, and hence each tree represents exactly $(k-j+1)$ categories from $\mc C_i$.
    Therefore,
    \begin{equation*}
      a_{ij}=\frac{(k-j+1)|\mc T_j|}{|\mc C_i|}=\frac{(k-j+1)(k-j)!j!}{(k-j+1)!}=j!,
    \end{equation*}
    which had to be shown.

    Case 3.
    Suppose $i>j$.
    The depth of all trees in $\mc T_j$ being $j$, they contain no nodes at distance $i$ from $v_0$.
    Hence none of the trees in $\mc T_j$ represents any categories from $\mc C_i$, therefore $a_{ij}=0$.
    This completes the proof of the lemma.
    \qed
  \end{proof}

  To define a laminar cover, assume that every laminar system represented by $\mc T_j$ is assigned the same coefficient $x_j$ in our solution for each $j\in\{1,\dots,k'\}$.
  That is, a solution to the following program will be given.
  \begin{subequations}
    \begin{align}\label{lp:simultan:laminarCoverLpPerType}
      \tag{LP8}
      \min&\sum_{j=1}^{k'}\frac{k!}{(k-j+1)!}x_j\\
      \mbox{s.t.}\quad\quad\quad\quad\quad&&\nonumber\\
      x&\in\Q_+^{k'}&\\
      Ax &\geq 1, &\label{eq:simultan:laminarCoverLpPerType:mtxIneq}
    \end{align}
  \end{subequations}
  where $A=(a_{ij})\in\Z^{k'\times k'}$ for the $a_{ij}$'s defined in Lemma~\ref{lem:simultan:laminCoveringNumber}, and $x_j$ is the common coefficient of all trees in $\mc T_j$ for $j\in\{1,\dots,k'\}$.
  Note that a solution $x$ to~(\ref{lp:simultan:laminarCoverLpPerType}) can be converted to a solution $y$ of~(\ref{lp:simultan:laminarCoverLp}) by setting coefficient $y_{\mtiny{\mc F}}$ of a laminar category system $\mc F\in\mc L_{\mc C}$ to $x_j$ if $\mc F$ is represented by one of the trees in $\mc T_j$ and to zero if no such $j$ exists (see the proof of Lemma~\ref{lem:simultan:optimalLaminarCover} for the exact argument).
  Now, we give a solution $\tilde x^+$ to~(\ref{lp:simultan:laminarCoverLpPerType}) and then show that the corresponding $y$ is an optimal solution to~(\ref{lp:simultan:laminarCoverLp}).
  \begin{lemma}
    Let
    \begin{equation}\label{eq:simultan:requSolToLES}
      \tilde x_j=
      \begin{cases}
        \frac{1}{k'!} & \text{if } j=k',\\
        (k-j-1)\tilde x_{j+1}+\frac{2j-k+1}{(j+1)!} & \text{otherwise}
      \end{cases}
    \end{equation}
    and
    \begin{equation}
      \tilde x^+_j=\max(\tilde x_j,0),
    \end{equation}
    where $j\in\{1,\dots,k'\}$.
    Then, $\tilde x^+$ is a feasible solution to~(\ref{lp:simultan:laminarCoverLpPerType}).
  \end{lemma}
  \begin{proof}
    Let $A=(a_{ij})\in\Z^{k'\times k'}$ be as above.
    First, we show that $A\tilde x=1$.
    The $i^{\text{th}}$ row of this linear equation system is
    \begin{equation}\label{eq:simultan:mtxMultiplicationDef}
      \sum_{j=1}^{k'}a_{ij}\tilde x_j=1.
    \end{equation}
    By substitution and simple rearrangement,~(\ref{eq:simultan:mtxMultiplicationDef}) is equivalent with
    \begin{equation}\label{eq:simultan:mtxMultiplicationRearranged}
      \tilde x_j=\frac{1-\sum_{q=j+1}^{k'}\frac{(k-j)!j!}{(k-q+1)!}\tilde x_q}{j!}
    \end{equation}
    for all $j\in\{1,\dots,k'\}$.
    We prove (\ref{eq:simultan:mtxMultiplicationRearranged}) by induction.
    For $j=k'$, (\ref{eq:simultan:mtxMultiplicationRearranged}) gives that $\tilde x_{k'}=\frac{1}{k'!}$, which matches the first case of~(\ref{eq:simultan:requSolToLES}).
    For the inductive step, assume that $j<k'$ and that (\ref{eq:simultan:mtxMultiplicationRearranged}) holds for $(j+1)$.
    Consider the following computation:
    \begin{multline}
      \tilde x_j
      =(k-j-1)\tilde x_{j+1}+\frac{2j-k+1}{(j+1)!}
      =\frac{j!(k-j-1)\tilde x_{j+1}+\frac{2j-k+1}{j+1}}{j!}\\
      =\frac{-j!\tilde x_{j+1}+\frac{2j-k+1}{j+1}+\frac{k-j}{j+1}(j+1)!\tilde x_{j+1}}{j!}\\
      =\frac{-j!\tilde x_{j+1}+\frac{2j-k+1}{j+1}+\frac{k-j}{j+1}(1-\sum_{q=j+2}^{k'}\frac{(k-j-1)!(j+1)!}{(k-q+1)!}\tilde x_q)}{j!}\\
      =\frac{1-j!\tilde x_{j+1}-\frac{k-j}{j+1}\sum_{q=j+2}^{k'}\frac{(k-j-1)!(j+1)!}{(k-q+1)!}\tilde x_q}{j!}\\
      =\frac{1-\frac{(k-j)!j!}{(k-(j+1)+1)!}\tilde x_{j+1}-\sum_{q=j+2}^{k'}\frac{(k-j)!j!}{(k-q+1)!}\tilde x_q}{j!}\\
      =\frac{1-\sum_{q=j+1}^{k'}\frac{(k-j)!j!}{(k-q+1)!}\tilde x_q}{j!},
    \end{multline}
    where the fourth and the last equation hold by induction.
    From this, $A\tilde x=1$ immediately follows.

    To complete the proof of the lemma, observe that $A\tilde x^+\geq 1$ holds, because $a_{ij}\geq 1$ and $A\tilde x=1$.
    Since $\tilde x^+$ is also non-negative, one gets that $\tilde x^+$ is a feasible solution to~(\ref{lp:simultan:laminarCoverLpPerType}).
    \qed
  \end{proof}

  We need the following property of the negative coordinates of $\tilde x$ before we further investigate $\tilde x^+$.
  \begin{lemma}\label{lem:simultan:negCoordsOfTheSolToLaminarCoverLpPerType}
    Let $\tilde x$ be as defined in~(\ref{eq:simultan:requSolToLES}).
    Then, for $j\in\{1,\dots,k'-1\}$, if $\tilde x_{j+1}<0$, then $\tilde x_{j}<0$.
  \end{lemma}
  \begin{proof}
    First we show that $\tilde x_{j+1}<0$ implies that
    \begin{equation}\label{eq:simultan:negCoordsOfTheSolToLaminarCoverLpPerType:indexIneq}
      j+1<\left\lceil\frac{k-1}{2}\right\rceil.
    \end{equation}
    If $k'<\left\lceil\frac{k-1}{2}\right\rceil$, then~(\ref{eq:simultan:negCoordsOfTheSolToLaminarCoverLpPerType:indexIneq}) follows, so one can assume that $k'\geq \left\lceil\frac{k-1}{2}\right\rceil$.
    It suffices to show that
    \begin{equation}\label{eq:simultan:negCoordsOfTheSolToLaminarCoverLpPerType:ineqSeq}
      0\leq \tilde x_{k'}\leq\dots\leq \tilde x_{\left\lceil\frac{k-1}{2}\right\rceil}
    \end{equation}
    holds when $j+1\geq\left\lceil\frac{k-1}{2}\right\rceil$, that is, if~(\ref{eq:simultan:negCoordsOfTheSolToLaminarCoverLpPerType:indexIneq}) does not hold.
    If $k'=\left\lceil\frac{k-1}{2}\right\rceil$, then (\ref{eq:simultan:negCoordsOfTheSolToLaminarCoverLpPerType:ineqSeq})~follows, since $\tilde x_{k'}=\frac{1}{k'!}\geq 0$.
    For $k'>\left\lceil\frac{k-1}{2}\right\rceil$, one gets that $\tilde x_{k'}\leq \tilde x_{k'-1}$, since
    \[\tilde x_{k'-1}=(k-k')\tilde x_{k'}+\frac{2k'-k-1}{k'!}=\frac{k'-1}{k'!}\geq\frac{1}{k'!}=\tilde x_{k'}\]
    holds when $k'\geq 2$.
    The next computation shows that $\tilde x_{q+1}\leq \tilde x_{q}$ for $\left\lceil\frac{k-1}{2}\right\rceil\leq q\leq k'-2$.
    \[\tilde x_{q}=(k-q-1)\tilde x_{q+1}+\frac{2q-k+1}{(q+1)!}\geq \tilde x_{q+1}+\frac{2q-k+1}{(q+1)!}\geq \tilde x_{q+1},\]
    where the first and the second inequality follows from $q\leq k'-2$ and $\left\lceil\frac{k-1}{2}\right\rceil\leq q$, respectively.
    This completes the proof of~(\ref{eq:simultan:negCoordsOfTheSolToLaminarCoverLpPerType:ineqSeq}), which in turn implies~(\ref{eq:simultan:negCoordsOfTheSolToLaminarCoverLpPerType:indexIneq}).

    Using~(\ref{eq:simultan:negCoordsOfTheSolToLaminarCoverLpPerType:indexIneq}), we show that $\tilde x_{j}<\tilde x_{j+1}$, which completes the proof of the lemma.
    \[\tilde x_{j}=(k-j-1)\tilde x_{j+1}+\frac{2j-k+1}{(j+1)!}<\tilde x_{j+1}+\frac{2j-k+1}{(j+1)!}\leq\tilde x_{j+1},\] where the first and the last inequality holds by $k-j-1\geq 2$ and $2j-k+1\leq 0$, respectively --- both of which follow by~(\ref{eq:simultan:negCoordsOfTheSolToLaminarCoverLpPerType:indexIneq}).
    \qed
  \end{proof}

  Now, an optimal solution $\tilde y$ to~(\ref{lp:simultan:laminarCoverLp}) will be constructed.
  \begin{lemma}\label{lem:simultan:optimalLaminarCover}
    For $\mc F\in\mc L_{\mc C}$, let
    \begin{equation}
      \tilde y_{\mtiny{\mc F}}=
      \begin{cases}
        \tilde x^+_j & \text{if $\mc T_j$ contains the tree representation of $\mc F$},\\
        0 & \text{otherwise.}
      \end{cases}
    \end{equation}
    Then, $\tilde y$ is an optimal solution to~(\ref{lp:simultan:laminarCoverLp}) and its objective value equals~(\ref{eq:simultan:alphaVal}).
  \end{lemma}
  \begin{proof}
    Note that $\tilde y_{\mtiny{\mc F}}$ is well-defined, since each laminar category system ${\mc F}$ is represented by a unique tree $T$ and any tree is contained in at most one of $\mc T_1,\dots,\mc T_k$.

    To see that $\tilde y$ is a feasible solution to~(\ref{lp:simultan:laminarCoverLp}), recall that $a_{ij}$ is the number of times the trees in $\mc T_j$ cover a category with an index of size $i$.
    Therefore,~(\ref{eq:simultan:laminarCoverLpPerType:mtxIneq}) immediately implies~(\ref{eq:simultan:laminarCoverLP:ineq}) for $C_I\in\mc C\setminus\{C_\emptyset\}$.
    To show that~(\ref{eq:simultan:laminarCoverLP:ineq}) also holds for $C_\emptyset$, observe that $C_\emptyset$ is contained in any maximal laminar category system, which implies that~(\ref{eq:simultan:laminarCoverLpPerType:mtxIneq}) holds for $C_\emptyset$ (since $k\geq 1$), but then~(\ref{eq:simultan:laminarCoverLP:ineq}) holds for $C_\emptyset$ as well.
    Hence $\tilde y$ is indeed a feasible solution to~(\ref{lp:simultan:laminarCoverLp}).
    To prove the optimality of $\tilde y$, we need the dual of~(\ref{lp:simultan:laminarCoverLp}):
    \begin{subequations}
      \begin{align}\label{lp:simultan:laminarCoverLpDual}
        \tag{LP7D}
        \max&\sum_{C\in\mc C}\pi_C\\
        \mbox{s.t.}\quad\quad\quad\quad\quad&&\nonumber\\
        \pi&\in\Q_+^{\mc C}&\\
        \sum_{C\in \mc F}\pi_C &\leq 1 &\forall \mc F\in\mc L_{\mc C}\label{eq:simultan:laminarCoverLPDual:ineq}
      \end{align}
    \end{subequations}

    By Lemma~\ref{lem:simultan:negCoordsOfTheSolToLaminarCoverLpPerType}, if $\tilde x_j<0$ for some $j\in\{1,\dots,k\}$, then $\tilde x_{j'}<0$ for all $j'\in\{1,\dots,j\}$.
    Let $t$ denote the largest index for which $\tilde x_t<0$, and let $t=0$ if no such index exists.
    Note that $t\in\{0,\dots,k'-1\}$.
    We show that $\tilde\pi\in\Q_+^{\mc C}$ is a feasible solution to~(\ref{lp:simultan:laminarCoverLpDual}), where
    \begin{equation}\label{eq:simultan:piDef}
      \tilde\pi_{C_I}=
      \begin{cases}
        \frac{1}{k-t} & \text{if } |I|>t,\\
        0 & \text{otherwise}\\
      \end{cases}
    \end{equation}
    for $C_I\in\mc C$. Let ${\mc F}$ be an arbitrary maximal laminar category system.
    By the tree representation, it is clear that there are at least $(t+1)$ categories in ${\mc F}$ with an index of size at most $t$, which means that ${\mc F}$ contains at most $k+1-(t+1)=k-t$ categories with an index of size larger than $t$.
    This implies that $\tilde\pi\in\Q_+^{\mc C}$ satisfies constraints~(\ref{eq:simultan:laminarCoverLPDual:ineq}), hence it is dual-feasible.

    In fact, $\tilde\pi$ satisfies all constraints in~(\ref{eq:simultan:laminarCoverLPDual:ineq}) with equality that correspond to a positive $\tilde y_{\mtiny{\mc F}}$ coordinate (these correspond to those laminar systems whose tree representation appear is some $\mc T_j$), since then ${\mc F}$ contains exactly $k+1-(t+1)=k-t$ categories with an index set of size more than $t$ and hence $\sum_{C\in {\mc F}}\tilde\pi_C = (k+1-(t+1))\frac{1}{k-t} = 1$.

    Observe that $\tilde x_j<0$ if and only if $(A\tilde x^+)_j>1$.
    Hence $\tilde x^+$ fulfills all constraints of~(\ref{eq:simultan:laminarCoverLpPerType:mtxIneq}) with equality except for the first $t$.
    Recall that $x_j=y_{\mtiny{\mc F}}$ for those laminar category systems ${\mc F}$ which are represented by one of the trees in $\mc T_j$, and $a_{ij}$ is the number of times the trees in $\mc T_j$ cover each category with an index of size $i$.
    This means that $\tilde y_{\mtiny{\mc F}}$ fulfills an inequality of~(\ref{eq:simultan:laminarCoverLP:ineq}) with equality if and only if the inequality corresponds to a category with an index of size larger than $t$.
    But, by~(\ref{eq:simultan:piDef}), a coordinate of $\tilde\pi_C$ is positive if and only if it corresponds to such a category.

    Therefore, $\tilde y$ and $\tilde\pi$ are complementary solutions, which implies that $\tilde y$ is indeed an optimal solution.
    It remains to show that the objective value of $\tilde y$ matches~(\ref{eq:simultan:alphaVal}).
    Let $g(j)=\frac{1}{k-j}\sum_{i=j+1}^{k'}\binom{k}{i}$ for $j\in\{0,\dots,k'-1\}$.
    Using this notation, we have to show that
    \begin{equation*}
      \sum_{\mc F\in\mc L_{\mc C}}\tilde y_{\mtiny{\mc F}}=\max_{j\in\{0,\dots,k'-1\}}g(j).
    \end{equation*}
    By the Duality theorem,
    \begin{equation}\label{eq:simultan:optCoverAndFij}
      \sum_{\mc F\in\mc L_{\mc C}}\tilde y_{\mtiny{\mc F}}=\sum_{C\in\mc C}\tilde\pi_C=\frac{1}{k-t}\sum_{i=t+1}^{k'}\binom{k}{i}=g(t),
    \end{equation}
    where $\tilde\pi$ is as defined in~(\ref{eq:simultan:piDef}).
    From this, one gets that $\max_{j\in\{0,\dots,k'-1\}}g(j)\geq g(t)=\sum_{\mc F\in\mc L_{\mc C}}\tilde y_{\mtiny{\mc F}}$.
    To prove equality, we show that $g(j)\leq g(t)$ holds for all $j\in\{0,\dots,k'-1\}$ by defining a feasible dual solution whose objective value is $g(j)$.
    Let $\tilde\pi'\in\Q_+^{\mc C}$ be such that
    \begin{equation*}
      \tilde\pi'_{C_I}=
      \begin{cases}
        \frac{1}{k-j} & \text{if } |I|>j,\\
        0 & \text{otherwise}\\
      \end{cases}
    \end{equation*}
    for $C_I\in\mc C$.
    Since $\tilde\pi'$ is a feasible dual solution, its objective value $\sum_{\mc F\in\mc L_{\mc C}}\tilde\pi'_C$ is at most that of the optimal dual solution $\tilde\pi$ --- the latter equals to $g(t)$ by~(\ref{eq:simultan:optCoverAndFij}).
    Hence
    \begin{equation*}
      g(t)\geq\sum_{\mc F\in\mc L_{\mc C}}\tilde\pi'_C=\frac{1}{k-j}\sum_{i=j+1}^{k'}\binom{k}{i}=g(j)
    \end{equation*}
    follows.
    Therefore, $\max_{j\in\{0,\dots,k'-1\}}g(j)=g(t)=\sum_{\mc F\in\mc L_{\mc C}}\tilde y_{\mtiny{\mc F}}$, which had to be shown.
    \qed
\end{proof}

  As noted above, $\tilde y$ corresponds to an $(m,\ell)$-cover $F_1,\dots,F_m$ for which $\alpha(k,k')=\frac{m}{\ell}$.
  To complete the proof of Theorem~\ref{thm:simultan:laminarCoverMain}, observe that all maximal laminar category systems can be enumerated in $\mathcal{O}(f(k))$ time, furthermore, one can find an optimal solution to the problem restricted to any of them in polynomial time by Theorem~\ref{thm:simultan:locLamCaseSolvable}.
  The heaviests among these solutions are at least as good as the bests among the optimal solutions to the problems restricted to $F_i$ for $i\in\{1,\dots,m\}$ --- which are $\alpha(k,k')$-approximate by Theorem~\ref{thm:simultan:apxComb}.
  \qed
\end{proof}

\begin{remark}
  The proof of Lemma~\ref{lem:simultan:categorysetsonly} also shows that the value of $\alpha(k,k')$ derived in Theorem~\ref{thm:simultan:laminarCoverMain} would remain the same even if we considered \emph{locally}-laminar $(m,\ell)$-covers, which are less restrictive than laminar $(m,\ell)$-covers.
\end{remark}

Table~\ref{tbl:simultan:apxRatios} summarizes the value of $\alpha(k,k')$ given by Theorem~\ref{thm:simultan:laminarCoverMain} for small values of $k$ and $k'$.
By Theorems~\ref{thm:simultan:apxComb}~and~\ref{thm:simultan:laminarCoverMain} we get the following.
\setlength{\extrarowheight}{2pt}
\begin{table}
  \centering
  \begin{tabularx}{4.9cm}{|Y|Y|Y|Y|Y|Y|}
    \hline
    \diagbox[innerleftsep=3pt,innerrightsep=2pt,width=.8cm,height=.8cm]{$k$}{$k'$}
    & $1$ & $2$ & $3$ & $4$ & $5$ \\ \hline
    $1$ & \cellcolor{gray!50}$1$& & & & \\ \hline
    $2$ & \cellcolor{gray!50}$1$& \cellcolor{gray!50}$\!\sfrac{3}{2}$& & & \\ \hline
    $3$ & \cellcolor{gray!50}$1$& $2$& $\!\sfrac{7}{3}$& & \\ \hline
    $4$ & \cellcolor{gray!50}$1$& $\!\sfrac{5}{2}$& $\!\sfrac{7}{2}$& $\!\!\sfrac{15}{4}$ & \\ \hline
    $5$ & \cellcolor{gray!50}$1$& $3$& $5$& $\!\!\sfrac{25}{4}$ & $\!\!\sfrac{13}{2}$ \\ \hline
  \end{tabularx}
  \caption{The approximation guarantees given by Theorem~\ref{thm:simultan:laminarCoverMain}, where $k=|\mc H|$ and every edge is in at most $k'$ subgraphs in $\mc H$.
    The highlighted values match the integrality gap of~(\ref{lp:simultan:degLpStrong}).}\label{tbl:simultan:apxRatios}
\end{table}
\begin{theorem}\label{thm:simultan:lamCombApx}
  One can find an $\alpha(k,k')$-approximate solution to the simultaneous assignment problem in $\mathcal{O}(f(k)\poly(|V|,|E|))$ time, where $k=|\mc H|$ and every edge appears in at most $k'$ of the subgraphs in $\mc H$.
\end{theorem}
For small $k$, this approximation guarantee is significantly better than that of the greedy algorithm, which is $(2k+1)$.

By Theorem~\ref{thm:simultan:locallyLaminarPolytope}, (\ref{lp:simultan:degLpStrong}) becomes integer when the problem is restricted to a locally laminar edge set, hence applying Theorems~\ref{thm:simultan:apxGapBound}~and~\ref{thm:simultan:laminarCoverMain}, one gets the following.
\begin{theorem}\label{thm:simultan:lamApxMain}
  The integrality gap of (\ref{lp:simultan:degLpStrong}) is at most $\alpha(k,k')$, where $k=|\mc H|$ and every edge appears in at most $k'$ of the subgraphs in $\mc H$.
\end{theorem}

We show that the integrality gap of~(\ref{lp:simultan:degLpStrong}) is exactly $\alpha(k,k')$ for $k\in\{1,2\}$.
For $k=1$, the integrality gap of~(\ref{lp:simultan:degLpStrong}) is one by Theorem~\ref{thm:simultan:locallyLaminarPolytope}.
Figure~\ref{fig:simultan:gapExamplek2} shows an example for $k=2$, where $x\equiv\frac{1}{2}$ is a feasible solution with objective value $\frac{3}{2}$, while the optimal integer solutions consist of a single edge.
It remains open if this bound is tight for any $k\geq 3$.

\begin{remark}
  The value of $\alpha(k,k')$ does not change if we do not require that $\mc L$ is laminar.
  Instead, it is enough to ensure that $\mc L$ becomes laminar when restricted to $F_i$, where $F_1,\dots,F_m$ is the laminar cover given by Theorem~\ref{thm:simultan:laminarCoverMain}.
\end{remark}

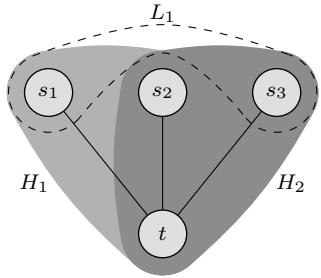
\begin{figure}
  \centering
  \begin{tikzpicture}[scale=.75]
    \SetVertexMath
    dot/.style = {circle, draw, semithick,
      inner sep=0pt, minimum size=3pt,
      node contents={}},
    \tikzset{edge/.style = {->,> = latex'}}
    \Vertex[x=-2,y=3.5,L=s_1]{s1}
    \Vertex[x=0,y=3.5,L=s_2]{s2}
    \Vertex[x=2,y=3.5,L=s_{3}]{s3}

    \Vertex[empty,x=0,y=4.7]{du}
    \Vertex[empty,x=0,y=4.2]{dd}

    \Vertex[x=0,y=1,L=t]{t}

    \draw[] (t) -- (s1);
    \draw[] (t) -- (s2);
    \draw[] (t) -- (s3);

    \begin{pgfonlayer}{background}
      \begin{scope}[opacity=.8,transparency group]
        \highlight{11mm}{black!30}{(s1.center) to [bend right =10] (t.center) to [bend right =10] (s2.center) to [bend right = 10] (s1.center)}
        \highlight{11mm}{black!30}{(s1.center) to [bend left =10] (t.center) to [bend left =10] (s2.center) to [bend left = 10] (s1.center)}
      \end{scope}
      \begin{scope}[opacity=.8,transparency group]
        \highlight{11mm}{black!45}{(s2.center) to [bend right =10] (t.center) to [bend right =10] (s3.center) to [bend right = 10] (s2.center)}
        \highlight{11mm}{black!45}{(s2.center) to [bend left =10] (t.center) to [bend left =10] (s3.center) to [bend left = 10] (s2.center)}
      \end{scope}
      \begin{scope}
        \draw[dashed, dash phase=1.4pt] (s1) ++(-230-22.5:20pt)coordinate(s11)  arc (-230-22.5:-40:20pt)   coordinate(s12);
        \draw[dashed, dash phase=.0pt] (s3) ++(-230+90:20pt)coordinate(s31)  arc (-230+90:-40+90+22.5:20pt)   coordinate(s32);
        \draw[dashed, dash phase=1.6pt] plot[smooth, tension=.7] coordinates {(s11) (du) (s32)};
        \draw[dashed, dash phase=1.6pt] plot[smooth, tension=.7] coordinates {(s12) (dd) (s31)};
      \end{scope}

    \end{pgfonlayer}
    \node at (-2.25,1.9){$H_1$};
    \node at (+2.25,1.9){$H_2$};
    \node at (0,4.9){$L_1$};
  \end{tikzpicture}
  \caption{Tight integrality gap example for $k=2$.
    Here the two subgraphs $H_1$ and $H_2$ in $\mc H$ are represented by the light and dark highlights, while the only set $L_1=\{s_1,s_3\}$ of $\mc L$ corresponds to the dashed area.
    We set $b\equiv 1$ and $c\equiv 1$.}\label{fig:simultan:gapExamplek2}
\end{figure}

\subsection{Improved Approximation Algorithm for Uniform $b$}\label{sec:simultan:apx:uniformBounds}
This section investigates the case of \emph{uniform} $b$, that is, when there exists $a\in\Z_+$ such that $b_{H}\equiv a$ for all $H\in\mc H$.
The following approach is yet another application of the approximation framework given at the beginning of this section.
\begin{theorem}
  If $b$ is uniform, then the integrality gap of~(\ref{lp:simultan:degLpStrong})~is at most $\frac{k+1}{2}$, and a $\frac{k+1}{2}$-approximate solution can be found in strongly polynomial time.
\end{theorem}
\begin{proof}
  Let $x_i$ denote an optimal solution to the problem restricted to those edges which are in either $H_i\in\mc H$ or non of the subgraphs in $\mc H$, and let $z$ denote an optimal solution to the problem restricted to those edges of $G$ which are included in at most one of the subgraphs in $\mc H$.
  Firstly, we show that~(\ref{lp:simultan:degLpStrong})~is integer for the problem restricted to $H_i$ and also that $x_i$ can be found in strongly polynomial time.
  In the problem restricted to $H_i$, all degree constraints posed in other subgraphs in $\mc H$ are redundant, hence we can delete all subgraphs other than $H_i$.
  The size of $\mc H$ being one, we conclude that $x_i$ can be found in strongly polynomial time and~(\ref{lp:simultan:degLpStrong})~is integer for the restricted problem by Theorems~\ref{thm:simultan:locLamCaseSolvable}~and~\ref{thm:simultan:locallyLaminarPolytope}.
  Secondly, if one restricts the problem to the edges contained in at most one of the subgraphs in $\mc H$, then the problem is again laminar.
  Therefore, $z$ can be found in strongly polynomial time and~(\ref{lp:simultan:degLpStrong})~is integer when the problem is restricted to these edges.

  Furthermore, observe that $z,x_1,\dots,x_k$ is a $(k+1,2)$-cover of $E$.
  By Theorem~\ref{thm:simultan:apxGapBound}, this implies that the integrality gap of~(\ref{lp:simultan:degLpStrong}) is at most $\frac{k+1}{2}$ and, by Theorem~\ref{thm:simultan:apxComb}, the heaviest among $z,x_1,\dots,x_k$ is a $\frac{k+1}{2}$-approximate solution, which completes the proof of the theorem.
  \qed
\end{proof}

As a special case, this gives a $\frac{3}{2}$-approximation algorithm for the weighted double matching problem.

\subsection{When the Graph is Sparse}\label{sec:simultan:apx:sparse}
Throughout this section, we assume that $\mc L=\emptyset$.
We say that a graph $G=(V,E)$ is \emph{$\frac{m}{\ell}$-sparse} if $\ell i_G(X)\leq m(|X|-1)$ holds for every $\emptyset\neq X\subseteq V$, where $i_G(X)$ denotes the number of edges induced by node set $X\subseteq V$.
The following classical result of Nash-Williams characterizes $\frac{m}{1}$-sparse graphs.
\begin{theorem}[Nash-Williams~\cite{NW64}]\label{thm:simultan:NashWilliams}
  The edge set of a graph $G=(V,E)$ can be covered by $m$ forests if and only if $i_G(X)\leq m(|X|-1) $ for every $\emptyset\neq X\subseteq V$.
\end{theorem}

From Theorem~\ref{thm:simultan:NashWilliams}, we derive the following characterization of $\frac{m}{\ell}$-sparse graphs.
\begin{corollary}\label{cor:simultan:multiCoverWithForests}
  There exists an $(m,\ell)$-cover $F_1,\dots,F_m\subseteq E$ consisting of forests if and only if $G$ is $\frac{m}{\ell}$-sparse, that is, $\ell i_G(X)\leq m(|X|-1)$ for every $\emptyset\neq X\subseteq V$.
\end{corollary}
\begin{proof}
  By definition, forests $F_1,\dots,F_m\subseteq E$ form an $(m,\ell)$-cover if every edge $e\in E$ is contained in at least $\ell$ of $F_1,\dots,F_m$.
  Without loss of generality, one can assume that every edge is contained in exactly $\ell$ of $F_1,\dots,F_m$.
  To reduce the problem to Theorem~\ref{thm:simultan:NashWilliams}, we replace every edge $e\in E$ with $\ell$ parallel edges.
  Let $G'=(V,E')$ denote the graph obtained this way.
  There exists an $(m,\ell)$-cover $F_1,\dots,F_m\subseteq E$ in $G$ which consists of forests if and only if there exist $m$ disjoint forests in $G'$ covering $E'$.
  By Theorem~\ref{thm:simultan:NashWilliams}, the latter holds if and only if $i_{G'}(X)\leq m(|X|-1)$ for every $\emptyset\neq X\subseteq V$.
  Since $i_{G'}(X)=\ell i_G(X)$ by the construction of $G'$, there exists an $(m,\ell)$-cover $F_1,\dots,F_m\subseteq E$ in $G$ if and only if $\ell i_G(X)\leq m(|X|-1)$.
  The latter means that $G$ is $\frac{m}{\ell}$-sparse, which completes the proof.
  \qed
\end{proof}

Using this characterization of the existence of $(m,\ell)$-covers consisting of forests, we derive several approximation algorithms for $\frac{m}{\ell}$ sparse graphs.

\begin{theorem}\label{thm:simultan:sparseGraphGapAndApx}
  If $G$ is $\frac{m}{\ell}$-sparse, $\mc L=\emptyset$ and $\mc H$ has the local-interval property, then the integrality gap of~(\ref{lp:simultan:degLp}) is at most $\frac{m}{\ell}$, and an $\frac{m}{\ell}$-approximate solution can be found in strongly polynomial time.
\end{theorem}
\begin{proof}
  By Corollary~\ref{cor:simultan:multiCoverWithForests}, if $\ell i_G(X)\leq m(|X|-1)$ for every $\emptyset\neq X\subseteq V$, then there exists an $(m,\ell)$ cover $F_1,\dots,F_m\subseteq E$ consisting of forests.
  By Theorem~\ref{thm:simultan:treeInterval}, if the problem is restricted to $F_i$, then the polyhedron defined by~(\ref{lp:simultan:degLp}) is integer for all $i\in\{1,\dots,m\}$.
  But then Theorem~\ref{thm:simultan:apxGapBound} implies that the integrality gap of~(\ref{lp:simultan:degLp}) is at most $\frac{m}{\ell}$.

  By Theorem~\ref{thm:simultan:treeInterval}, one can efficiently solve the simultaneous assignment problem for trees when $\mc H$ has the local-interval property.
  Therefore, Theorem~\ref{thm:simultan:apxComb} immediately implies that an $\frac{m}{\ell}$-approximate solution can be found in strongly polynomial time.
  \qed
\end{proof}

\begin{corollary}
  If $G$ is a planar graph with minimum cycle length $m$, $\mc L=\emptyset$ and $\mc H$ has the local-interval property, then the integrality gap of~(\ref{lp:simultan:degLp}) is at most $\frac{m}{m-2}$ and an approximate solution can be found in strongly polynomial time with the same guarantee.
\end{corollary}
\begin{proof}
  It suffices to show that $G$ is $\frac{m}{m-2}$-sparse, because then Theorem~\ref{thm:simultan:sparseGraphGapAndApx} can be applied.
  Let $e,f$ and $n$ denote the number of edges, the number of faces and the number of nodes, respectively.
  By Euler's formula, $f-e+n=2$ holds, and $mf\leq 2e$ follows since the length of the shortest cycle is $m$.
  But then $2m=mf-me+mn\leq 2e-me+mn$ and hence $(m-2)i(X)\leq m|X|$ for every $\emptyset\neq X\subseteq V$.
  This means that $G$ is $\frac{m}{m-2}$-sparse, and hence the proof is complete by Theorem~\ref{thm:simultan:sparseGraphGapAndApx}.
  \qed
\end{proof}

\begin{corollary}
  If $G$ is a Laman graph, $\mc L=\emptyset$ and $\mc H$ has the local-interval property, then the integrality gap of~(\ref{lp:simultan:degLp}) is at most $2$ and an approximate solution can be found in strongly polynomial time with the same guarantee.
\end{corollary}

\begin{theorem}\label{thm:simultan:evenCycleIntegralPolytope}
  If $G$ is an even cycle and $\mc L = \emptyset$, then~(\ref{lp:simultan:degLp})~is integer and the simultaneous assignment problem is solvable in strongly polynomial time.
\end{theorem}
\begin{proof}
  For each node of the even cycle, one can omit all but the most restrictive bound posed in a subgraph in $\mc H$.
  This means that the problem is a capacitated $b$-matching problem on an even cycle --- with infinite bound at node $v$ if the two edges incident to $v$ do not appear in the same subgraphs in $\mc H$.
  It is well known that the capacitated $b$-matching polyhedron is described by the natural constraints for bipartite graphs, therefore,~(\ref{lp:simultan:degLp})~is integer and the problem can be solved in strongly polynomial time.
  \qed
\end{proof}

Figure~\ref{fig:simultan:pseodoTreeGapEg1} shows that Theorem~\ref{thm:simultan:evenCycleIntegralPolytope} does not hold even if we add only a leaf to a 4-cycle.

\begin{theorem}
  Assume that $\mc H$ has the local-interval property and $\mc L=\emptyset$.
  If $G$ is a cactus graph with minimum cycle length $m$, then the integrality gap of~(\ref{lp:simultan:degLp}) is at most $\frac{m}{m-1}$.
\end{theorem}
\begin{proof}
  It is easy to see that there exist $m$ forests $F_1,\dots,F_{m}$ which cover every edge at least $(m-1)$-times.
  \qed
\end{proof}

\begin{theorem}\label{thm:simultan:pseudoForestGap}
  Assume that $\mc H$ has the local-interval property, $\mc L=\emptyset$ and that $G$ is a pseudo-tree with cycle length $m$.
  Then the integrality gap of~(\ref{lp:simultan:degLp}) is at most
  \[\begin{cases}
    \frac{m+1}{m} & \mbox{if $m$ is even},\\
    \frac{m}{m-1} &\mbox{if $m$ is odd,}
  \end{cases}\]
  and an approximate solution with the same guarantee can be found in strongly polynomial time.
\end{theorem}
\begin{proof}
  Let $C\subseteq E$ be the cycle and let $F_1,\dots,F_m$ denote the trees obtained by leaving out the edges of $C$ one at a time.
  Clearly, every edge $e\in E$ is contained in at least $(m-1)$ of $F_1,\dots,F_m$ and the polyhedron defined by~(\ref{lp:simultan:degLp}) is integer when the problem is restricted to $F_i$, hence the integrality gap of~(\ref{lp:simultan:degLp}) is at most $\frac{m}{m-1}$ by Theorem~\ref{thm:simultan:apxGapBound}.
  For even $m$, we improve this bound as follows.

  Observe that every edge $e\in E$ is contained in at least $m$ of $C,F_1,\dots,F_m$.
  Since~(\ref{lp:simultan:degLp}) defines an integer polyhedron when the problem is restricted to $C$, one gets that the integrality gap of~(\ref{lp:simultan:degLp}) is at most $\frac{m+1}{m}$ by Theorem~\ref{thm:simultan:apxGapBound}, which completes the proof.
  \qed
\end{proof}

Figures~\ref{fig:simultan:pseodoTreeGapEg1}~and~\ref{fig:simultan:pseodoTreeGapEg2} show tight examples with even and odd cycle length for Theorem~\ref{thm:simultan:pseudoForestGap}, respectively.

\begin{figure}[h!]
  \centering
  \begin{subfigure}[t]{.49\textwidth}
    \centering
    \begin{tikzpicture}[scale=.9]
      \centering
      \SetVertexMath
      dot/.style = {circle, draw, semithick,
        inner sep=0pt, minimum size=3pt,
        node contents={}},
      \tikzset{edge/.style = {->,> = latex'}}
      \begin{scope}[rotate=45]
        \Vertex[x=-1, y=1,L={v_2}]{v_2}
        \Vertex[x=1, y=1,L={v_1}]{v_1}
        \Vertex[x=-1, y=-1,L={v_3}]{v_3}
        \Vertex[x=1, y=-1,L={v_4}]{v_4}
        \Vertex[x=2.4142, y=-2.4142,L={v_5}]{v_5}

        \draw[] (v_2) -- (v_1) node [midway, above=-2pt, opacity=100] {$ $};
        \draw[] (v_4) -- (v_1) node [midway, right=-2pt, opacity=100] {$ $};
        \draw[] (v_4) -- (v_3) node [midway, below=-2pt, opacity=100] {$ $};
        \draw[] (v_2) -- (v_3) node [midway, left=-2pt, opacity=100] {$ $};
        \draw[] (v_4) -- (v_5) node [midway, left=-2pt, opacity=100] {$ $};

        \begin{pgfonlayer}{background}
          \begin{scope}[opacity=.8,transparency group]
            \highlight{11mm}{black!30}{(v_3.center) to [bend left =10] (v_2.center) to [bend left =10] (v_1.center) to [bend left =10] (v_4.center) to [bend right = 0] (v_5.center)}
          \end{scope}
          \begin{scope}[opacity=.8,transparency group]
            \highlight{8mm}{black!45}{(v_1.center) to [bend right =10] (v_2.center) to[bend right =10] (v_3.center) to[bend right =10] (v_4.center) to [bend left = 0] (v_5.center)}
          \end{scope}
          \begin{pgfonlayer}{background}
          \end{pgfonlayer}
        \end{pgfonlayer}
        \node at (0,-1.84){$H_2$};
        \node at (+1.95,0){$H_1$};
      \end{scope}
    \end{tikzpicture}
    \caption{}\label{fig:simultan:pseodoTreeGapEg1}
  \end{subfigure}
  \hfill
  \begin{subfigure}[t]{.49\textwidth}
    \centering
    \begin{tikzpicture}[scale=.9]
      \centering
      \SetVertexMath
      dot/.style = {circle, draw, semithick,
        inner sep=0pt, minimum size=3pt,
        node contents={}},
      \tikzset{edge/.style = {->,> = latex'}}
      \begin{scope}[rotate=-45]
        \Vertex[x=-1, y=1,L={v_1}]{v_1}
        \begin{scope}[rotate=72]
          \Vertex[x=-1, y=1,L={v_2}]{v_2}
        \end{scope}
        \begin{scope}[rotate=2*72]
          \Vertex[x=-1, y=1,L={v_3}]{v_3}
        \end{scope}
        \begin{scope}[rotate=3*72]
          \Vertex[x=-1, y=1,L={v_4}]{v_4}
        \end{scope}
        \begin{scope}[rotate=4*72]
          \Vertex[x=-1, y=1,L={v_5}]{v_5}
        \end{scope}

        \draw[] (v_2) -- (v_1) node [midway, above=-2pt, opacity=100] {$ $};
        \draw[] (v_5) -- (v_1) node [midway, right=-2pt, opacity=100] {$ $};
        \draw[] (v_4) -- (v_3) node [midway, below=-2pt, opacity=100] {$ $};
        \draw[] (v_2) -- (v_3) node [midway, left=-2pt, opacity=100] {$ $};
        \draw[] (v_4) -- (v_5) node [midway, left=-2pt, opacity=100] {$ $};

        \begin{pgfonlayer}{background}
          \begin{scope}[opacity=.8,transparency group]<
            \highlight{11mm}{black!30}{(v_3.center) to [bend left =10] (v_2.center) to [bend left =10] (v_1.center) to [bend left =10] (v_5.center) to [bend right = 0] (v_4.center)}

          \end{scope}
          \begin{scope}[opacity=.8,transparency group]
            \highlight{8mm}{black!45}{(v_1.center) to [bend right =10] (v_2.center) to[bend right =10] (v_3.center) to[bend right =10] (v_4.center) to [bend left = 0] (v_5.center)}
          \end{scope}
          \begin{pgfonlayer}{background}
          \end{pgfonlayer}
        \end{pgfonlayer}
        \begin{scope}[rotate=2*72+31+5]
          \node at (-1.4,1.4){$H_2$};
        \end{scope}

        \begin{scope}[rotate=-31-5]
          \node at (-1.5,1.5){$H_1$};
        \end{scope}
      \end{scope}
    \end{tikzpicture}
    \caption{}\label{fig:simultan:pseodoTreeGapEg2}
  \end{subfigure}
  \caption{Let $H_1$ and $H_2$ be the light and dark subgraphs, respectively.
    For both instances, $x\equiv\frac{1}{2}$ is a feasible fractional solution, while the optimal integer solution consists of two edges.
    Hence, the integrality gap is $\frac{5}{4}$ in both cases.}\label{fig:simultan:pseodoTreeGapEg}
\end{figure}
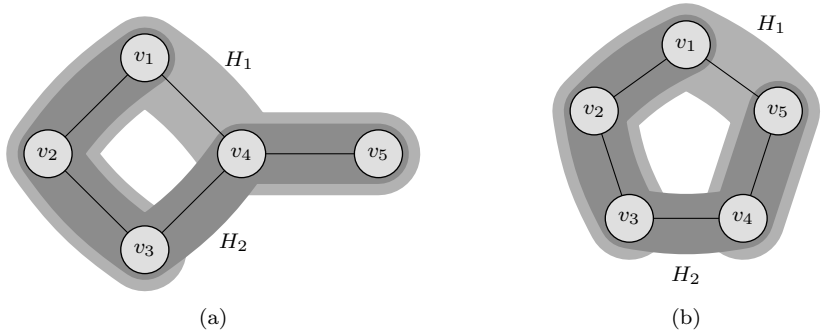

\section{Open Questions}\label{sec:simultan:open}
The bounded-violation algorithm for the upper bounded degree $g$-polymatroid element problem, mentioned in the introduction, can be improved in the special case $|\mc H|=1$ and $\mc L=\emptyset$, since then we get back the $b$-matching problem, which can be solved efficiently without any violation of the degree-constraints --- instead of degree-violation $2|\mc H|-1=1$ provided by the algorithm.
It is not clear if the degree-violation $(2|\mc H|-1)$ can be improved when $|\mc H|\geq 2$.
It is also open whether the approach extends to the case when $\mc L\neq \emptyset$.

It would be interesting to identify further polynomial-time solvable special cases.
Maybe this would also enable the application of the framework described in Section~\ref{sec:simultan:apx:constAPX} to derive new approximation algorithms.

By Theorem~\ref{cor:simultan:tree2Subgr}, the problem is solvable in polynomial time when the graph is a tree, $|\mc H|=2$ and $\mc L=\emptyset$.
A natural next step would be to investigate the case when $|\mc H|=3$ and $\mc L=\emptyset$.

By Theorem~\ref{thm:simultan:lamApxMain}, we have an upper bound on the integrality gap of~(\ref{lp:simultan:degLpStrong}), which is shown to be tight when $|\mc H|=2$.
It remains open whether this upper bound is tight for larger $\mc H$ or the upper bound can be further improved.

We leave it open if Theorem~\ref{thm:simultan:sapAPXhard2} holds for the unweighted case, that is, whether the case $|\mc H|=2$, $\mc L = \emptyset$ is hard to approximate in the unweighted case for bipartite graphs.
If it remains hard, then the proof of Theorem~\ref{thm:simultan:wdmapxc} may also imply that the unweighted distance matching problem cannot be arbitrarily approximated either.

In the \emph{degree-prescribed subgraph problem}, we are given a graph $G=(V,E)$ and degree-prescription $D_v\subseteq\{1,\dots,\deg(v)\}$ for each $v\in V$.
The question is whether there exists a subgraph $M$ of $G$ such that the $M$-degree of each node $v\in V$ is in $D_v$.
This problem can be solved in polynomial time for a wide range of inputs --- for example, when no two consecutive integers are left out from any $D_v$~\cite{lovasz1972factorization,generalFactorsOfGraphs}
or the size of $D_v$ is at least $\left\lceil\frac{\deg(v)}{2}\right\rceil$ for all $v\in V$~\cite{FrankDegreeConstrained}.
This problem being a generalization of the $b$-matching problem, it is quite natural to ask if the laminar simultaneous assignment problem can be generalized along these lines such that the problem remains tractable.
For example, instead of the degree-sum constraint~(\ref{eq:simultan:degLp:card}) for $\mc L$, we could consider a prescription $D_L$ for all $L\in\mc L$.
The case of parity prescription, when $D_v$ consists of either all even or all odd numbers, is already an interesting question.

\section{Acknowledgment}
The author is grateful to Kristóf Bérczi for his support and for suggesting relevant literature. The author also thanks Alpár Jüttner his remarks on the presentation of the results.\\
The work was supported by the Lend\"ulet Programme of the Hungarian Academy of Sciences -- grant number LP2021-1/2021.\\
Supported by the \'UNKP-22-4 New National Excellence Program of the Ministry for Culture and Innovation from the source of the National Research, Development and Innovation Fund.

\bibliographystyle{spmpsci}      
\bibliography{bibliography}   

\end{document}